\documentclass{article}
\usepackage[hidelinks]{hyperref}
\usepackage[utf8]{inputenc}
\usepackage{amsmath}
\usepackage{amsthm}
\usepackage{cleveref}
\usepackage{amsfonts}
\usepackage[numbers,sort]{natbib}
\usepackage{float}
\usepackage{array}
\usepackage{dirtytalk}
\usepackage{xcolor}
\usepackage{thm-restate}
\usepackage{graphicx}
\usepackage{amsmath}
\usepackage{bbm}
\usepackage{fullpage}
\usepackage{pgfplots}
\pgfplotsset{compat=1.17}
\usepackage[shortlabels]{enumitem}
\usetikzlibrary{patterns, decorations.pathmorphing, arrows.meta}
\usetikzlibrary{decorations.pathreplacing}
\usepackage{enumitem}

\DeclareMathOperator*{\argmin}{arg\,min}

\newcommand\blfootnote[1]{%
  \begingroup
  \renewcommand\thefootnote{}%
  \NoHyper\footnote{#1}\endNoHyper
  \addtocounter{footnote}{-1}%
  \endgroup
}

\usepackage{arydshln}

\allowdisplaybreaks
\newtheorem{theorem}{Theorem}[section]
\newtheorem*{theorem*}{Theorem}
\newtheorem{lemma}[theorem]{Lemma}

\newtheorem{definition}[theorem]{Definition}
\newtheorem{remark}[theorem]{Remark}

\theoremstyle{definition}

\newcommand{\opt}{\textsc{Opt}}
\newcommand{\socialopt}{\textsc{Social\text{-}Opt}}
\newcommand{\panelopt}{\textsc{Panel\text{-}Opt}}
\newcommand{\cost}{\textsc{Cost}}

\newcommand{\panelcore}{\textsc{Panel\text{-}Core}}
\newcommand{\socialcore}{\textsc{Social\text{-}Core}}
\newcommand{\social}{\textsc{Social\text{-}Cost}}
\newcommand{\panelcost}{\textsc{Panel\text{-}Cost}}

\newcommand{\uniformpanel}{\mathcal{U}_{k,n}}

\newcommand{\budgets}{\mathcal{X}_{B}}

\newcommand{\pbinstance}{\langle m, B, (\cost_1,\ldots,\cost_n) \rangle}
\newcommand{\facilityinstance}{\langle \mathcal{X},\mathcal{C}, d, (x_1,\ldots,x_n) \rangle}
\newcommand{\committeeinstance}{\langle \mathcal{X},\mathcal{C},d, \ell, (x_1,\ldots,x_n) \rangle} %

\title{The Panel Complexity of Sortition: Is 12 Angry Men Enough?
}
\author{ 
Johannes Brustle$^\ast$
\quad
Simone Fioravanti$^\dagger$
\quad 
Tomasz Ponitka$^\ddagger$
\quad
Jeremy Vollen$^\mathsection$
}
\date{April 29, 2025}

\begin{document}

\maketitle

\blfootnote{We thank the audience at the Theory of Computer Science seminar at the University of Illinois Urbana-Champaign for their helpful comments during the presentation of an earlier version of this paper. This project has been partially funded by MIUR PRIN 2022, Learning in Markets and Society, by the European Research Council (ERC) under the European Union's Horizon 2020 research and innovation program (grant agreement No. 866132), by PNRR MIUR project PE0000013-FAIR, by an Amazon Research Award, by the Israel Science Foundation Breakthrough Program (grant No. 2600/24), by a grant from the TAU Center for AI and Data Science (TAD), and by the NSF-BSF (grant number 2020788).}
\blfootnote{$^\ast$Sapienza University of Rome, Italy. Email: \texttt{brustle@diag.uniroma1.it}}
\blfootnote{$^\dagger$Sapienza University of Rome, Italy. Email: \texttt{fioravanti@diag.uniroma1.it}}
\blfootnote{$^\ddagger$Tel Aviv University, Israel. Email: \texttt{tomaszp@mail.tau.ac.il}}
\blfootnote{$^\mathsection$University of New South Wales, Australia. Email: \texttt{j.vollen@unsw.edu.au}}

\begin{abstract}
Sortition is the practice of delegating public decision-making to randomly selected panels. 
Recently, it has gained momentum worldwide through its use in citizens' assemblies, sparking growing interest within the computer science community.
One key appeal of sortition is that random panels tend to be more \emph{representative} of the population than elected committees or parliaments.
Our main conceptual contribution is a novel definition of {representative panels}, based on the Wasserstein distance from statistical learning theory.
Using this definition, we develop a framework for analyzing the \emph{panel complexity} problem—determining the required panel size to ensure desirable properties.
We focus on three key desiderata: (1) that efficiency at the panel level extends to the whole population, measured by social welfare; (2) that fairness guarantees for the panel translate to fairness for the population, captured by the core; and (3) that the probability of an outlier panel, for which the decision significantly deviates from the optimal one, remains low.
We establish near-tight panel complexity guarantees for these desiderata across two fundamental social choice settings: participatory budgeting and facility location.
\end{abstract}

\newpage

\section{Introduction}

A citizens' assembly is a panel of randomly selected citizens convened to deliberate on a policy question and collectively express a position.
While the practice dates back millennia, it is not lacking in recent examples, including the Irish Citizens’ Assembly (2016–2018), for which 100 randomly selected citizens addressed key legal and policy issues—such as abortion, aging, and climate change—leading to a national referendum on abortion and a government-declared climate emergency. 
Many other citizens' assemblies have been convened worldwide on a range of policy topics \cite{oecd2020innovative}.
A defining feature of citizens' assemblies is their use of sortition---random selection of participants---an idea that has gained increasing attention in the computer science community \citep[e.g.,][]{BenadeGP19,FlaniganGGP20,FlaniganGGPG2021,FlaniganKP21,FLPW24a,halpern2024federatedassemblies, BaharavF24, WalshX12, DBLP:conf/icml/CaragiannisM024, DBLP:conf/nips/EbadianKMP022,MeirST21,Cheng2018Distortion,ChengDK17, DBLP:journals/corr/abs-2406-00913}.

An important question in sortition is how to determine the optimal panel size. 
Naturally, the logistical and financial burden of a citizens' assembly increases with its size. 
At the same time, smaller panels are 
less likely to represent the entire population accurately.
Moreover, other applications of sortition, such as jury selection, offer little guidance on navigating this trade-off. 
In the U.S., juries almost always consist of 12 members---a tradition dating back to 725, when King Morgan of Gla-Morgan, Wales, is said to have chosen this number to mirror the 12 Apostles \cite{Gorski2012}.
The 
central topic 
of this paper is the \emph{panel complexity} problem---an effort to establish a solid
foundation for determining the optimal panel size.\footnote{While some prior works have implications for panel complexity, to the
 best of our knowledge, our work is the first to systematically study the panel complexity problem. We contrast our results with prior work in \Cref{sec:related}.}

To precisely formulate the panel complexity problem, we must first establish a criterion for evaluating panels.
Much of the literature on sortition focuses on the design and analysis of sortition algorithms with respect to two key desiderata:  \emph{representativeness} and {\emph{selection fairness}}.
Broadly, representativeness refers to ensuring that a citizens' assembly reflects the diverse perspectives of the underlying population \citep{Youn02a,Will00a}, while selection fairness {is the ideal that posits that all citizens should have an equal chance of participation.}
In this work, we focus  on the case where the panel members are selected uniformly at random from the population,\footnote{{We elaborate on and further motivate this assumption in more detail below}.} ensuring selection fairness. 
Thus, our analysis is primarily concerned with achieving representativeness.

Much of the prior work defines representative panels based on compliance with exogenous quotas for demographic features \citep{FlaniganGGP20,FlaniganGGPG2021,FlaniganKP21,FLPW24a}.
However, specifying quotas---and even determining which features should have quotas---may be difficult before resolving the question of panel size.
Representativeness has also been defined in terms of the expected loss of a random panel of a given size relative to the optimal panel of the same size \citep{DBLP:conf/nips/EbadianKMP022}, or through the notion of the (approximate) core \citep{DBLP:journals/corr/abs-2406-00913}.
However, these definitions are not suitable for our context, as they do not provide a meaningful way to compare the representativeness of panels of different sizes.
Thus, we ask:

\vspace{0.1cm}
     \noindent \textbf{Question 1:} \emph{What is the right measure of panel representativeness for analyzing panel complexity?}
        \vspace{0.1cm}

Sortition-selected panels are often tasked with concrete decisions, such as legal or policy matters.
We draw inspiration from \citet*{Fish18a}, who argues that a citizens’
assembly's decision should approximate what would emerge if the entire population deliberated. The panel size required to reach this end will necessarily depend on the particulars of the policy decision at hand, such as the outcome space and the structure of preferences. Hence, it is essential to study panel complexity in the context of specific {\em social choice settings}.

A recent work of \citet*{DBLP:conf/icml/CaragiannisM024} provided a tight characterization of the panel complexity of approximately maximizing social welfare in the {\em facility location} problem within a finite metric space.
In this model, we are given a finite metric space, and each individual in the population has an ideal point in this space. The objective is to identify a single point (facility) that best serves the entire population.
While this is a promising result, it has two key limitations: (1) it does not account for scenarios with infinitely many alternatives, such as facility location in the Euclidean space; and (2) it does not capture preferences that go beyond distances in a metric space,
such as those in the {\em participatory budgeting} problem \citep{FainGM16,GoelKSA19}, where a fixed budget must be distributed across a set of potential projects, and individuals differ in their preferences over how much funding each project should receive.
Hence, we ask:

\vspace{0.1cm}
     \noindent \textbf{Question 2:} \emph{What is the panel complexity in social choice settings beyond facility location in finite spaces?}
     
      \vspace{0.1cm}
Moreover, \citet*{DBLP:conf/icml/CaragiannisM024}
assume that panels {always} converge to a decision that maximizes utilitarian welfare among their members. However, this assumption may be unrealistic for two reasons:
(1) reaching a welfare-maximizing decision may be infeasible due to time constraints or the strategic behavior of panel members. A large body of research on facility location examines strategy-proof mechanisms, and it is well-known that such mechanisms cannot achieve optimal welfare, only approximate it \cite{chan2021mechanismdesignfacilitylocation}; and
(2) even if the panel can agree on the welfare-maximizing alternative, this choice may not be fair to all members, which is a central concern in the participatory budgeting literature \cite{FainGM16}.
These concerns motivate the study of alternative forms of {\em deliberation guarantees}---such as approximate welfare maximization or core fairness---for the decision reached by the sortition-selected panels.
Therefore, we ask:

\vspace{0.1cm}
     \noindent \textbf{Question 3:} \emph{What is the panel complexity for deliberation guarantees beyond exact welfare maximization?}
        \vspace{0.1cm}

In this work, we address all three of the questions above.

{We note that} the social choice settings we study involve
{individual preferences}---such as cost functions over budget allocations in participatory budgeting---which may not be easily observable or verifiable, unlike demographic features such as age or gender.
For this reason, we focus our analysis on uniform {panel} selection rather than 
{on}
algorithms 
{that}
use 
{individual}
preference 
information to select panels.
The use of uniform {panel} selection naturally rules out the possibility of guarantees that hold for every panel in the support of the lottery---that is, ex-post guarantees; see \Cref{sec:related} for more details---but it improves the robustness of our results by ensuring they do not rely on access to any 
{individual's} information.
An open direction for future work is to extend our model to allow non-uniform panel selection via demographic quotas, and to analyze how this affects panel complexity; see \Cref{sec:conclusion}.

\subsection{Our Contributions}
    We consider a sortition model in which a panel of $k$ agents is selected uniformly at random from a population of $n$ agents.
Each of our panel complexity guarantees consists of four components: (i) the social choice setting, (ii) the deliberation guarantee, (iii) panel complexity, and (iv) the population guarantee. 
That is, given a specific social choice setting, assuming the panel's decisions satisfy the deliberation guarantee and the panel size meets the required threshold, we show that the corresponding population guarantee holds for the entire population
if the panel is drawn uniformly at random.
Before presenting our panel complexity guarantees for participatory budgeting and facility location, we first introduce our novel definition of representativeness, which is the key tool for deriving the panel complexity bounds.

    \paragraph{Representative Panels.}
    Our first conceptual contribution is the definition of $\epsilon$-representative panels with respect to a given feature $f: [n] \to \mathcal{X}$, where $(\mathcal{X},d)$ is a metric space. A panel is $\epsilon$-representative if the distribution it induces over the values of $f$ is at most $\epsilon$ away, in Wasserstein distance (also known as the earth mover’s distance; see \Cref{def:Wdistance}), from the distribution induced by the entire population.
This is a natural notion of representativeness and is particularly well-suited for panel complexity, unlike prior definitions in the literature \citep{FlaniganGGP20,FlaniganGGPG2021,FlaniganKP21,FLPW24a,DBLP:conf/nips/EbadianKMP022,DBLP:journals/corr/abs-2406-00913}. 

Wasserstein distance is a well-established measure in the study of sample complexity.
However, unlike the typical setting in which sample complexity is analyzed---where samples are independent and identically distributed (i.i.d.), corresponding to sampling {\em with replacement}---sortition involves sampling {\em without replacement}, where panels are drawn uniformly at random from all size-$k$ subsets of the population.

We establish two general lemmas (\Cref{lem:W1expectation,lem:W1concentration}) for adapting standard sample complexity techniques to panel complexity analysis.
Using these tools, we establish that a panel of size $O((1/\epsilon)^2 \cdot (\log \ell + \log(1/\delta)))$ is $\epsilon$-representative with probability $1-\delta$ for any sequence of $\ell$ real-valued features (\Cref{thm:many_features}).  
Notably, this guarantee is independent of both the population size and the specific features. This theorem is instrumental in deriving upper bounds on panel complexity for participatory budgeting and facility location. We also believe it may be useful for analyzing panel complexity in other social choice problems.

    \paragraph{Participatory Budgeting.}
    We study the (divisible) participatory budgeting problem, where the goal is to allocate a budget $B$ across $m$ projects. Each agent in the population has a Lipschitz cost function over the space of possible budget allocations $[0,1]^m$.  
    Lipschitz cost functions capture the idea that agents' preferences change smoothly with the allocation.
    This class {is quite general} and includes important examples, such as linear costs of the form $\sum_{j=1}^m \alpha_{i,j} \cdot (1-x_j)$ for allocation $x \in [0,1]^m$ and agent-specific  $\alpha_{i,1}, \ldots, \alpha_{i,m} \geq 0$. 

Our first result shows that if the panel selects an efficient allocation among its members, it remains efficient for the entire population.

\begin{theorem*}[Panel Complexity for Welfare Maximization in Participatory Budgeting; see \Cref{thm:pbeffub}] \; 
\begin{enumerate}[noitemsep, topsep=0pt,label=(\roman*)] 
    \item {\em Social Choice Setting:} Consider the participatory budgeting problem with Lipschitz cost functions.
     \item {\em Deliberation Guarantee:} Suppose the panel-selected budget allocation minimizes the total cost among panel members up to an additive factor $\tau \geq 0$ and a constant multiplicative factor $\rho \geq 1$.
    \item {\em Panel Complexity:} 
     Suppose the panel size $k$ is $O\left((1/\epsilon^2) \cdot m \cdot \log(1/\epsilon)\right)$.
     \item {\em Population Guarantee:} Then, the expected social cost is optimal up to an additive factor $\tau + \epsilon$ and a multiplicative factor $\rho$.
 \end{enumerate}
\end{theorem*}

The panel complexity guarantee is essentially tight in two respects. First, we derive a nearly matching lower bound of $\Omega((1/\epsilon^2) \cdot m)$ for minimizing social cost (\Cref{thm:pb_lower_bound}). Second, we show that no panel size can yield a multiplicative approximation error without incurring any additive approximation error  (\Cref{thm:pb_impossibility}).

Our second result shows that if the panel's allocation is fair among its members, this fairness extends to the entire population.

\begin{theorem*}[Panel Complexity for Core Fairness in Participatory Budgeting; see \Cref{thm:core_ub}] \; 
\begin{enumerate}[noitemsep, topsep=0pt,label=(\roman*)]
    \item {\em Social Choice Setting:} Consider the participatory budgeting problem with Lipschitz cost functions.
     \item {\em Deliberation Guarantee:} Suppose the panel-selected  allocation satisfies core fairness for panel members up to a budget fraction $\eta \geq 0$, an additive factor $\tau \geq 0$, and a constant multiplicative factor $\rho \geq 1$.
    \item {\em Panel Complexity:} 
    Suppose the panel size $k$ is $O\left((1/\epsilon^2) \cdot (m \cdot \log(1/\epsilon) + \log(1/\delta))\right)$.
     \item {\em Population Guarantee:} Then, with probability at least $1 - \delta$, 
     the budget allocation satisfies 
     core fairness for the population
     up to a budget fraction $\eta + \epsilon$, an additive factor $\tau + \epsilon$, and a multiplicative factor $\rho$.
 \end{enumerate}
\end{theorem*}

    \paragraph{Facility Location.}
We also study the facility location problem, where the goal is to place a single facility in a metric space $\mathcal{X}$. Each agent in the population is located at some point in $\mathcal{X}$ and seeks to minimize their distance to the facility.  

Our first result shows that the probability of a panel selecting a facility location far from the optimal one is small. Notably, this result holds for arbitrary, possibly infinite, metric spaces. 

\begin{theorem*}[Panel Complexity for Bounding Panel Outliers in Facility Location; see \Cref{thm:tail_bound}] \; 
\begin{enumerate}[noitemsep, topsep=0pt,label=(\roman*)]
    \item {\em Social Choice Setting:} Consider the facility location problem in an arbitrary metric space.
     \item {\em Deliberation Guarantee:} Suppose the panel-selected  facility minimizes total distance to panel members.
    \item {\em Panel Complexity:} Suppose  the panel size $k$ is $O\left(\log(1/\delta)\right)$.
     \item {\em Population Guarantee:} Then, for any constant $T > 2$, with probability at least $1 - \delta$,
     the facility location is within distance $T \cdot \opt$ of the optimal, where $\opt$ is the optimal social cost.
 \end{enumerate}
\end{theorem*}

We further show that the above guarantee does not hold for $T = 2$ (\Cref{prop:star_lower_bound}).

Our second result is analogous to the efficiency guarantees established for participatory budgeting. Importantly, in the facility location setting, we achieve a multiplicative error guarantee {for any space with bounded dimension; see \Cref{def:assouad} for a formal definition of the notion of dimension that we use.}

\begin{theorem*}[Panel Complexity for Welfare Maximization in Facility Location; see \Cref{thm:assouad}] \; 
\begin{enumerate}[noitemsep, topsep=0pt,label=(\roman*)]
    \item {\em Social Choice Setting:} Consider the facility location problem in a metric space with Assouad dimension $t$, which includes the Euclidean space $[0,1]^t$ under any $\ell_p$-norm.
     \item {\em Deliberation Guarantee:} Suppose the panel-selected facility  minimizes total distance to panel members.
    \item {\em Panel Complexity:} Suppose  the panel size $k$ is $O\left((1/\epsilon^2) \cdot t \cdot \log(1/\epsilon)\right)$.
     \item {\em Population Guarantee:} Then, the expected social cost is optimal up to a multiplicative factor of $1+\epsilon$.
 \end{enumerate}
\end{theorem*}

This bound is almost tight: we also establish a lower bound of $\Omega((1/\epsilon^2) \cdot t)$ for $\mathcal{X} = ([0,1]^t, \ell_{\infty})$ (\Cref{thm:assouad_lower}),
showing that the panel complexity guarantee above is essentially optimal.

Finally, we extend our analysis to the case of locating multiple facilities. Assuming the panel minimizes total distance to its members, we show that in any metric space $\mathcal{X}$, the resulting loss in social cost is upper bounded by the Wasserstein distance between the panel and population distributions (\Cref{lem:multifacilities}). This result leads to the following panel complexity guarantee for social cost minimization over the line $\mathcal{X} = [0,1]$. Notably, this guarantee is independent of the number of facilities to be selected.

\begin{theorem*}[Panel Complexity for Multiple Facility Location on the Line; see \Cref{thm:line}] \; 
\begin{enumerate}[noitemsep, topsep=0pt,label=(\roman*)]
    \item {\em Social Choice Setting:} Consider the multiple facility location problem in $[0,1]$.
     \item {\em Deliberation Guarantee:} Suppose the panel-selected facilities minimize total distance to panel members.
    \item {\em Panel Complexity:} Suppose  the panel size $k$ is  $O\left(1/\epsilon^2\right)$.
     \item {\em Population Guarantee:} Then, the expected social cost is optimal up to an additive factor of $\epsilon$.
 \end{enumerate}
\end{theorem*}

\subsection{Related Work}\label{sec:related}

\paragraph{Panel Complexity.}
    Several works have established panel complexity-type results, though none focus on it as their primary objective.
    \citet*{DBLP:conf/icml/CaragiannisM024} and \citet*{goel2025metric} both provide a panel complexity bound in the facility location setting.
    Unlike the result of \citet*{DBLP:conf/icml/CaragiannisM024}, the result of \citet*{goel2025metric} does not depend on the number of alternatives, but their model differs significantly from our own: panels can be sampled from a continuum of agents an infinite number of times and agents collectively reveal a ranking over all alternatives.
    
   In the setting of centroid clustering, \citet*{CFLM19} give a result which can be interpreted as a panel complexity result with respect to fairness and their bound depends on the number of candidate centroids.
    When every point is a candidate centroid, clustering is equivalent to the setting of \Cref{sec:multifacilities}.
    In this setting, \citet*{MiSh20} give a panel complexity bound for $\ell_2$ distances.
    Extending this result to general metric spaces remains open; see \Cref{sec:conclusion}.
    
\paragraph{Sortition.}
    The computer science literature studying sortition consists of two main threads. Starting with {the work of} \citet*{BenadeGP19}, the first thread \citep{BaharavF24,FLPW24a,FlaniganGGP20,FlaniganGGPG2021,FlaniganKP21} focuses on algorithms which select a panel from a pool of random invitees who have agreed to participate. 
    In the presence of selection bias, these selection algorithms then pursue representativeness by complying with quotas while also carefully balancing agents' selection probabilities. 
    Other works refine these algorithms by taking into consideration other desiderata such as manipulation-robustness and transparency \citep{BaharavF24,FLPW24a}.

    The other major thread of sortition research \citep{DBLP:conf/icml/CaragiannisM024,ChengDK17,Cheng2018Distortion,DBLP:conf/nips/EbadianKMP022,DBLP:journals/corr/abs-2406-00913,MeirST21,goel2025metric}, which our work falls within, takes the perspective that sortition should sample directly from the entire population \citep{gastil2019legislature}.
    Within this line of work, a number of papers have looked at the expected distortion---taking the ratio of expected social cost to the optimal social cost---of a decision made by a sortition-selected panel \citep{DBLP:conf/icml/CaragiannisM024,MeirST21,anagnostides2022metric,goel2025metric}.
    Rather than quota compliance, works in this strain of literature give novel definitions of representativeness.
    The definitions of \citet*{DBLP:conf/nips/EbadianKMP022} and \citet*{DBLP:journals/corr/abs-2406-00913} both frame representativeness in terms of the latent preferences of the population of agents. 
    To do this, they parameterize their definitions by $q$, assuming that a panel incurs a cost to each agent equal to their distance from the $q$-th closest panel member.
    \citet*{DBLP:conf/nips/EbadianKMP022} then take a distortion-like approach to measuring representativeness. 
    \citet*{DBLP:journals/corr/abs-2406-00913} give a property inspired more by the notion of proportional representation, which can be thought of as belonging to a broader class of proportional representation properties for centroid clustering \citep{aziz2023proportionally,kellerhals2025proportional, kalayci2024proportional}.
    Our notion is better-suited for evaluating panel size and does not reason about agents' preferences over panels.
    We comment further on how our approach to representativeness differs from theirs in \Cref{sec:representative}.

    We note that while uniform panel selection studied in this work is a standard and well-established model of sortition  \citep{DBLP:conf/icml/CaragiannisM024,ChengDK17,Cheng2018Distortion,DBLP:conf/nips/EbadianKMP022,DBLP:journals/corr/abs-2406-00913,MeirST21,goel2025metric}, prior work has also considered non-uniform selection methods---such as the Fair Greedy Capture method proposed by \citet*{DBLP:journals/corr/abs-2406-00913}---which are designed to provide ex-post guarantees that hold for every panel in the support of the lottery.
    For instance, \citet*{DBLP:conf/icml/CaragiannisM024} also achieved an ex-post constant-factor approximation to optimal social welfare in  facility location  using Fair Greedy Capture, which assumes full access to agents’ locations in the metric space.
    Although uniform panel selection inherently rules out the possibility of such ex-post guarantees, it has the important advantage of requiring no additional information about the population.

        \paragraph{Sample Complexity and Social Choice.}
        
    The literature on sample complexity in metric social choice settings has primarily focused on proving distortion guarantees for a constant number of samples \citep{DBLP:conf/wine/FainGMS17,DBLP:conf/aaai/FainGMP19,DBLP:conf/icalp/GoyalSSG23}. 
    In contrast, \citet*{DBLP:conf/ijcai/FainFM20} give results for an arbitrary number of samples $k$. They give constant bounds for higher moments of distortion. This approach demonstrates the value of larger panels but does not readily offer a way of selecting a panel size.
    Interestingly, similar to our \Cref{thm:tail_bound}, they give a result showing that their mechanism deviates from a constant approximation to the optimum with probability decaying exponentially in the number of samples. 
    Despite the similarity, their result follows from Markov's inequality, which does not suffice in our case.
    We note that our results have a similar flavor to those of \citet*{DBLP:conf/atal/DeyB15} and \citet*{DBLP:conf/aaai/CaragiannisF22}. 
    However, these works focus on voting settings with ordinal and approval-based preferences, respectively, in addition to a finite number of alternatives.
    
    In this paper, we study the sortition panel complexity of participatory budgeting, facility location, and multiple facility location. 
    In particular, 
    we study a participatory budgeting setting in which a divisible resource is allocated among projects with no fixed costs, commonly referred to as divisible participatory budgeting
    (for a survey, see \cite{AzSh21a}). 
    This is in contrast to indivisible participatory budgeting, in which funding decisions are binary \citep{rey2023computational}.
    Facility location and its multiple facilities extension are both classical problems in operations research \citep[e.g.,][]{owen1998strategic,tamir2001k}) and have been the topic of numerous recent works (see \cite{chan2021mechanismdesignfacilitylocation} for a survey).

\section{Framework for Panel Complexity}\label{sec:framework}

In this section, we develop our framework for analyzing panel complexity, which we later use to derive panel complexity guarantees for participatory budgeting (\Cref{sec:pb}) and facility location (\Cref{sec:facility}). The following three subsections introduce the sortition model, present a novel definition of representative panels for establishing upper bounds on panel complexity, and define camouflaged populations, which we use to derive lower bounds. The omitted proofs are deferred to \Cref{sec:missingthree}.

\subsection{Model of Sortition} 

For any $\ell \in \mathbb{N}$, let $[\ell] = \{1, \ldots, \ell\}$ denote the set of the first $\ell$ natural numbers. We represent a \emph{population} of $n$ agents as $[n]$ and refer to any subset $S \subseteq [n]$ as a \emph{panel}.

We consider features $f : [n] \to \mathcal{X}$, which assign each agent $i \in [n]$ a value in a metric space $(\mathcal{X}, d)$. 
Equivalently, we can regard $f$ as a vector $(f_1, \ldots, f_n) \in \mathcal{X}^n$.
These features may capture demographic characteristics or encode agents' preferences. For example, a feature $f : [n] \to [0,1]$ may represent agent $i$'s utility for a given budget allocation in a participatory budgeting problem. Similarly, a feature $f : [n] \to \mathcal{X}$ may represent agent $i$'s location in a facility location problem instance.

We study the properties of randomly selected panels of size $k$, where $1 \leq k \leq n$. Let $\uniformpanel$ denote the uniform distribution over the space $\binom{n}{k}$ of such panels. Throughout this work, we assume that panels $S \sim \uniformpanel$ are drawn uniformly at random.

\subsection{Representative Panels}\label{sec:representative}
Consider a single feature $f: [n] \to \mathcal{X}$, where $(\mathcal{X},d)$ is a metric space. This feature defines a discrete probability distribution over $\mathcal{X}$, determined by the frequency of each outcome. Formally, for each subset $S \subseteq [n]$, we define:
\begin{align*}
\quad \phi_f^S = \frac{1}{|S|} \sum_{i \in S} \delta(f_i).
\end{align*}
Here, $f_i$ denotes the value of feature $f$ assigned to agent $i$, and $\delta(f_i)$ represents the Dirac mass at $f_i$, i.e., the distribution $\phi_f^S$ assigns a probability mass of $\ell/|S|$ 
to each $x \in \mathcal{X}$, where $\ell$ is the number of occurrences of $x$ in the vector $(f_i)_{i \in S}$.
We refer to:
\begin{center}
    $\phi_f^{[n]}$ as the \emph{population distribution} of $f$; and  $\phi_f^S$, for a panel $S \subseteq [n]$, as the \emph{panel distribution} of $f$.
\end{center}

The Wasserstein distance, also known as the earth mover's distance or the Kantorovich-Rubinstein distance, originates from optimal transport theory~\cite{villani2008optimal, chewi2024statisticaloptimaltransport}. 
In general, the Wasserstein-$p$ distances form a family of metrics that generalize the notion of optimal transport. In this work, we focus exclusively on the Wasserstein-1 distance, which we simply refer to as the Wasserstein distance.
The Wasserstein distance is based on the idea that two probability measures are close if one can be transformed into the other by moving a small amount of probability mass.  It is defined as follows:

\begin{definition}[Wasserstein Distance]\label{def:Wdistance}
Let $\phi$ and $\psi$ be discrete probability distributions supported on a metric space $(\mathcal{X}, d)$, with finite supports. The Wasserstein distance between $\phi$ and $\psi$ is given by  
\begin{equation}
    W(\phi, \psi) = \min_{\gamma \in \Gamma(\phi, \psi)} \mathbb{E}_{(x,y) \sim \gamma} [d(x,y)] = \min_{\gamma \in \Gamma(\phi, \psi)} \sum_{x \in \operatorname{supp}(\phi)} \sum_{y \in \operatorname{supp}(\psi)} \gamma(x, y) \cdot d(x, y),
\end{equation}
where $\Gamma(\phi, \psi)$ is the set of all \emph{couplings} of $\phi$ and $\psi$, i.e., {$\Gamma(\phi, \psi)$ includes every joint distribution $\gamma$ over $\mathcal{X}^2$ 
such that
$\sum_{y \in \operatorname{supp}(\psi)} \gamma(x, y) = \phi(x)$ for all $x \in \operatorname{supp}(\phi)$ and $\sum_{x \in \operatorname{supp}(\phi)} \gamma(x, y) = \psi(y)$ for all $y \in \operatorname{supp}(\psi)$.}  
\end{definition}

The Wasserstein distance has applications in computer vision and machine learning, where it is used to quantify the similarity between images based on pixel frequencies~\cite{WGAN17, NiBCE09, RuzonT01}.  

We now define our notion of representativeness based on the Wasserstein distance.

\begin{definition}[Representative Panels]\label{def:representativeness}
    We say that a panel $S \subseteq [n]$ is $\epsilon$-representative of a given feature $f: [n] \to \mathcal{X}$ if the Wasserstein distance between the population distribution $\phi_f^{[n]}$ and the panel distribution $\phi_f^S$ is at most $\epsilon$, i.e., it holds that:
    \begin{align*}
        W(\phi_f^{[n]}, \phi_f^S) \leq \epsilon.
    \end{align*}
\end{definition}

For instance, panel distributions close to the population distribution include those obtained by moving each point  by a distance of less than $\epsilon$ or by shifting an $\epsilon$-fraction of all points by a distance of 1.

Prior work introduced alternative notions of panel representativeness~\cite{DBLP:journals/corr/abs-2406-00913,DBLP:conf/nips/EbadianKMP022}. \citet*{DBLP:conf/nips/EbadianKMP022} compare a given panel of size $k$ to the most representative panel of $k$ agents, while \citet*{DBLP:journals/corr/abs-2406-00913} define a criterion inspired by the ideal of proportional representation in committee voting. Both approaches rely on assumptions about agents' preferences over panels.  
In contrast, our notion of representativeness takes a more direct approach, using the underlying population as a benchmark. Crucially, this benchmark remains unchanged as the panel size varies, enabling an analysis of panel complexity.
{As we further illustrate in \Cref{sec:examplesthree}, our notion of representativeness is consistent with the intuition behind Example 7 of \citet*{DBLP:journals/corr/abs-2406-00913}; see \Cref{rem:sumdist}.}

In the rest of this section, we prove general results showing that if $k$ is sufficiently large, a randomly selected panel of size $k$ is $\epsilon$-representative for one or more features with high probability. From a technical perspective, it is useful to note the following basic fact: drawing a panel $S \sim \uniformpanel$ is equivalent to sampling $k$ individuals {\em without replacement} from the population $[n]$. 

Since concentration bounds are typically easier to derive for independent samples drawn {\em with replacement}, we begin by establishing a reduction for the expected Wasserstein distance in this setting.
The following lemma formalizes the intuition that sampling without replacement more closely reflects the underlying population than sampling with replacement; see \Cref{remark:withwithout} for a more detailed discussion.
{The proof follows from the convexity of  Wasserstein distance (\Cref{lem:w_convexity}) and the argument of \citet*{Hoeffding63}.}

\begin{restatable}[Reduction to Independent Sampling]{lemma}{SamplingWithoutReplacementExpectation}\label{lem:W1expectation}
   Let $f: [n] \to \mathcal{X}$ be a feature. 
Define $\mathcal{R}_{k,n}$ as the probability distribution over multisets of size $k$ drawn from $[n]$ by sampling with replacement. Then, we have:
\begin{align*}
    \mathbb{E}_{S \sim \uniformpanel} \left[ W(\phi_f^{[n]}, \phi_f^S) \right] 
    \leq \mathbb{E}_{S \sim \mathcal{R}_{k,n}} \left[ W(\phi_f^{[n]}, \phi_f^{S}) \right].
\end{align*}
\end{restatable}
We also provide a counterexample to a stronger version of the above lemma, demonstrating that the (first-order) stochastic dominance condition does not hold in general; see \Cref{ex:sd}.

In addition, we establish the following concentration inequality for the Wasserstein distance. {The proof relies on the bounded differences property of the Wasserstein distance by \citet*{weed2017sharpasymptoticfinitesamplerates} and a McDiarmid-style inequality for sampling without replacement by \citet*{sambale2022concentration}.}

\begin{restatable}[Concentration Inequality]{lemma}{SamplingWithoutReplacementConcentration}\label{lem:W1concentration}
   Let $f: [n] \to \mathcal{X}$ be a feature. Suppose that $d(x,y) \leq 1$ for all $x,y \in \mathcal{X}$.
Then, we have:
\begin{align*}
    \mathbb{P}_{S \sim \uniformpanel} \left[ W(\phi_f^{[n]}, \phi_f^S) \geq \mu + t \right] 
    \leq \exp\left( {-t^2 k} / 4 \right) \quad\text{where }\mu = \mathbb{E}_{S \sim \uniformpanel}\left[W(\phi_f^{[n]}, \phi_f^S)\right].
\end{align*}
\end{restatable}

With these two lemmas, deriving panel complexity bounds for $\epsilon$-representativeness of a feature valued in a general metric space reduces to bounding the expected Wasserstein distance for independent samples (drawn with replacement). We apply this approach below for real-valued features; see the proof of \Cref{thm:many_features}.

\paragraph{Real-Valued Features.}
We now focus on the special case of $\mathcal{X} = [0,1]$  with the Euclidean distance.

The next theorem addresses the panel complexity of being $\epsilon$-representative for a given sequence of features $f_1, \ldots, f_{\ell}$, and will be instrumental in establishing our panel complexity bounds.

\begin{restatable}[Panel Complexity of Representativeness]{theorem}{PanelMultipleFeatures}\label{thm:many_features}
Let $f_1, \ldots, f_{\ell}$ be features such that $f_i: [n] \rightarrow [0,1]$ for all $i \in [\ell]$. Let the panel size be:
    \begin{align*}
        k = O\left( \left( \frac{1}{\epsilon} \right)^2 \left( \log \ell + \log\left( \frac{1}{\delta}\right) \right)\right).
    \end{align*}
   Then it holds that:
    \begin{align*}
        \mathbb{P}_{S \sim \uniformpanel} \left[ \text{$S$ is $\epsilon$-representative of $f_j$ for all $j \in [\ell]$} \right] \geq 1-\delta.
    \end{align*}
\end{restatable}
\begin{proof}
    For any feature $f : [n] \to [0,1]$, we denote:
    \begin{align*}
       \mu_f =  \mathbb{E}_{S \sim \uniformpanel}\left[W(\phi_f^{[n]}, \phi_f^S)\right] \quad \text{and}\quad  \mu'_f =  \mathbb{E}_{S \sim \mathcal{R}_{k,n}}\left[W(\phi_f^{[n]}, \phi_f^S)\right].
    \end{align*}
   For any feature $f : [n] \to [0,1]$, by \cite[Proposition 2.1]{chewi2024statisticaloptimaltransport}, we have $\mu'_f \leq \epsilon/2$ for $k = O(1/\epsilon^2)$.
By \Cref{lem:W1expectation}, it follows that $\mu_f \leq \epsilon/2$.
Therefore, by \Cref{lem:W1concentration}, 
we conclude that for any $j \in [\ell]$, it holds that:
\begin{align*}
        \mathbb{P}_{S \sim \uniformpanel} \left[ \text{$S$ is not $\epsilon$-representative of $f_j$} \right] \leq \mathbb{P}_{S \sim \uniformpanel} \left[ W(\phi_{f_j}^{[n]}, \phi_{f_j}^S) \geq \mu_{f_j} + \epsilon/2 \right] \leq \exp\left( \frac{-(\epsilon/2)^2 k}{4} \right) \leq \delta/\ell.
    \end{align*}
    The result follows by the union bound.
\end{proof}

Finally, our notion of representativeness ensures that the (true) mean of a feature $f$ over the population $[n]$ is close to its (empirical) mean over any $\epsilon$-representative panel $S$. This simple fact will be useful in proving results in the following sections. 

\begin{restatable}[Empirical Mean of Representative Panels]{lemma}{WassersteinImpliesAvg}\label{lem:wasserstein_average}
 Let $f:[n] \rightarrow [0,1]$ be a feature and $S \subseteq [n]$ be an $\epsilon$-representative panel of $f$. Then, it holds that $|(1/n) \cdot \sum_{i \in [n]} f(i) - (1/|S|) \cdot \sum_{i \in S} f(i)| \leq \epsilon$.  
\end{restatable}

We note that one could alternatively consider a weaker notion of representativeness than \Cref{def:representativeness}, by defining a panel to be $\epsilon$-representative of a real-valued feature $f: [n] \to [0,1]$ if the assertion of \Cref{lem:wasserstein_average} holds. This weaker notion would suffice to establish some of our panel complexity guarantees in the following sections. Moreover, under this weaker definition, the panel complexity guarantee of \Cref{thm:many_features} can be proved directly using Hoeffding's inequality for sampling without replacement.

However, our more general definition is significantly more powerful, as it enables us to extend the panel complexity guarantees for facility location to the problem of locating multiple facilities, as discussed in \Cref{sec:multifacilities}. For this result, the weaker notion is insufficient.
This suggests that our more general framework may be an essential tool for the analysis of settings beyond those considered in this paper. 
Furthermore, the stronger notion---defined via the Wasserstein distance---more naturally captures representativeness and we believe it to be of independent interest, regardless of its role in proving our panel complexity results.

\subsection{Camouflaged Populations}\label{sec:camouflaged}

In this section, we develop the necessary tools to establish our lower bounds on panel complexity.
We begin by defining {camouflaged populations}---populations that are nearly uniform with respect to a certain feature.
As we show below, the difficulty of distinguishing camouflaged populations with a small panel makes them a powerful tool for lower bounding panel complexity.

\begin{definition}[Camouflaged Populations]\label{def:almost_uniform}
    Fix integers $h,w,r \geq 2$.
    For each vector $z \in \{-1,+1\}^{h}$, we define a population of size $n = 2hwr$ and 
    a feature $f(z,h,w,r) : [n] \to \{1, \ldots, 2h\}$ as follows.
    We first choose any partition $(A_1, \ldots, A_{2h})$ of $[n]$ such that for each $j \in [h]$ it holds that $|A_{2j}| = r \cdot (w + z_j)$ and $|A_{2j-1}| = r \cdot (w-z_j)$.
    Then, we define $f_i(z, h, w, r) = j$ for all $i \in A_j$ and $j \in [2h]$.
\end{definition}

   For a panel $S = \{i_1, \ldots, i_k\}$, we denote $f_S(z, h, w, r) = (f_{i_1}(z, h, w, r), \ldots, f_{i_k}(z, h, w, r))$.
   Note that a $(1 + z_j/w)/(2h)$ fraction of the population is assigned value $2j$, while a $(1 - z_j/w)/(2h)$ fraction is assigned value $2j - 1$, totaling a $1/h$ fraction with values in $\{2j - 1, 2j\}$ for each $j \in [h]$.

The following lemma shows that a panel of size $k$ smaller than $\Omega(h \cdot w^2)$ drawn from a camouflaged population \emph{cannot} recover $(3/4) \cdot h$ entries of $z$ with probability $6/7$. The proof follows from a reduction to a standard sample complexity lower bound for sampling with replacement, as shown in \cite[Lemma 5.10]{duetting2025pseudodimensioncontracts}.

\begin{restatable}[Panel Complexity Lower Bound]{lemma}{panelcamouflaged}\label{lem:sample_lower_bound}
Fix integers $h,w \geq 2$.
There exists an integer $r \geq 1$ such that, if there exists some function 
$g : \{1,2,\ldots,{2h}\}^k \to \{-1,+1\}^{h}$ for some $k \geq 1$ satisfying:
\[
\mathbb{P}_{S \sim \uniformpanel} \left[ \| g(f_S(z,h,w,r)) - z \|_1 \leq h/4 \right] \geq 6/7
\quad \text{for all } z \in \{-1,+1\}^{h},
\]
then it must be that $k = \Omega(h \cdot w^2)$.
\end{restatable}

\section{Participatory Budgeting via Sortition}\label{sec:pb}

In this section, we study the participatory budgeting problem. We first define the model, then establish panel complexity guarantees related to efficiency in terms of social welfare, and finally analyze fairness in terms of the core. The proofs omitted in this section are deferred to \Cref{sec:missingfour}.

\subsection{Model of Participatory Budgeting}

In this section, we describe the (divisible) participatory budgeting model \cite[see, e.g.,][]{GoelKSA19,FainGM16}. 
An instance of the participatory budgeting problem is a tuple $\pbinstance$ where $m \geq 2$ denotes the number of projects, $B \in \mathbb{R}_{> 0}$ denotes the available budget, and $\cost_i : [0,1]^m \to [0,1]$ is the cost function of the $i$-th agent, mapping budget allocations to the agent's disutility. Our goal is to select a budget allocation $x = (x_1, \ldots, x_m) \in [0,1]^m$ such that $\sum_{i=1}^m x_i \leq B$. 
We denote by $\budgets = \{x \in [0,1]^m \,:\, \sum_{j=1}^m x_j \leq B\}$ the space of all feasible budget allocations.
We make two important assumptions on the agent's cost functions:
    \begin{enumerate}[noitemsep, topsep=0pt, label=\textit{(\roman*)}]
    \item \emph{Monotonicity}: $\cost_i(x) \geq \cost_i(y)$ for any $x, y\in [0,1]^m$ such that $x_j \leq y_j$ for every $j \in [m]$, and
    \item {\em Lipschitz continuity}: $|\cost_i(x) - \cost_i(y)| \leq \|x-y\|_1$ where $\|x-y\|_1 = \sum_{j=1}^m |x_j-y_j|$.
\end{enumerate}

{Given a panel of size $k$ with Lipschitz cost functions, a 
\emph{panel decision function}
$\overline{x}$ selects a feasible budget allocation, i.e., $\overline{x} : \mathcal{L}^k \to \budgets$  where $\mathcal{L} = \mathrm{Lip}([0,1]^m, [0,1])$ denotes the space of all Lipschitz functions from $[0,1]^m$ to $[0,1]$.
For simplicity of notation, we will denote the output of the panel decision function $\overline{x}$, given the cost functions of a panel of voters $S$, by $\overline{x}(S):= \overline{x}((\cost_i)_{i\in S})$.}

\subsection{Welfare Maximization in Participatory Budgeting}

In this section, we establish fundamental limits on the efficiency of participatory budgeting when decisions are made by a randomly selected panel rather than the full population.

We first introduce the following notation:
\begin{definition}[Social Cost for Participatory Budgeting]
We define:
    \begin{enumerate}[noitemsep, topsep=0pt, label=(\roman*)]
    \item the social cost of $x \in [0,1]^m$ is denoted by $\social(x) = (1/n) \cdot \sum_{i=1}^n \cost_i(x)$,
    \item the optimal social cost is denoted by $\socialopt = \min_{x \in \budgets} \social(x)$,
    \item the panel cost of $x \in [0,1]^m$ is denoted by $\panelcost(x,S) = (1/k) \cdot \sum_{i \in S} \cost_i(x)$, and
    \item the optimal panel cost is denoted by $\panelopt(S) = \min_{x \in \budgets} \panelcost(x,S)$.
\end{enumerate}
\end{definition}

The following theorem provides an upper bound on the panel size required to ensure a near-optimal budget allocation. 
Specifically, it shows that if the panel size is sufficiently large, then selecting the best allocation based on panel preferences guarantees that the expected social cost is close to the optimal population-wide allocation.

\begin{theorem}[Panel Complexity Upper Bound]\label{thm:pbeffub} \;
Fix any instance $\pbinstance$ of the participatory budgeting problem.
    Let the panel size be:
    \begin{align*}
        k = O\left( \left( \frac{1}{\epsilon}\right)^2 \left( m \log \left( \frac{1}{\epsilon} \right) \right) \right).
    \end{align*}
    Suppose that panel $S \subseteq [n]$ of size $k$ selects a budget allocation $\overline{x}(S)$ satisfying:
    \begin{align*}
        \panelcost(\overline{x}(S), S) \leq   \rho  \cdot\panelopt(S) + \tau.
    \end{align*}
    for some $\rho \geq 1$ and $\tau \geq 0$, where $\rho$ is a constant. Then it holds that:
    \begin{align*}
        \mathbb{E}_{S \sim \uniformpanel} \left[ \social(\overline{x}(S)) \right] \leq \rho \cdot \socialopt + \tau + \epsilon.
    \end{align*}
\end{theorem}
\begin{proof}
   There exists an $\epsilon/(12\rho)$-covering $(x_1, \ldots, x_w)$ of $\budgets$
   with size $w = O((1/\epsilon)^m)$; see \Cref{def:covering} for a formal definition of a covering of a metric space. We define $w$ features $f_j: N \to [0,1]$ for $j \in \{1, \ldots, w\}$ by setting $f_j(i) = \cost_i(x_j)$. By \Cref{thm:many_features}, for $k = O((1/\epsilon)^2 \cdot (\log w + \log(1/\epsilon)) = O((1/\epsilon)^2 \cdot (m \log (1/\epsilon)))$, with probability at least $1 - \epsilon/2$, the panel $S \sim \uniformpanel$ is $(\epsilon/12)$-representative for all features $f_1, \ldots, f_w$. Fix any $(\epsilon/12)$-representative panel $S$.

   We first claim that for any budget allocation $x \in [0,1]^m$, it holds that:
    \begin{align} \label{claim}
        \big|\social(x) -   \panelcost(x,S)  \big| \leq \epsilon/(4\rho).
    \end{align}
   Indeed, by the definition of $\epsilon$-covering, there exists some $1 \leq j \leq w$ such that $\|x - x_j\|_1 \leq \epsilon$. Thus:
    \begin{align*}
       &  \big| (1/n) \cdot \sum_{i=1}^n \cost_i(x) - (1/|S|) \cdot \sum_{i \in S} \cost_i(x)\big| \\
        &\leq\big| (1/n) \cdot \sum_{i=1}^n \cost_i(x_j) - (1/|S|) \cdot \sum_{i \in S} \cost_i(x_j)\big|  + \epsilon/(6\rho) && (\text{since $\cost_i$ is Lipschitz}) \\
        &= \big| (1/n) \cdot \sum_{i=1}^n f_j(i)  - (1/|S|) \cdot \sum_{i \in S} f_j(i)\big|+ \epsilon/(6\rho) && (\text{by definition of $f_j(i)$}) \\
        &\leq \epsilon/(4\rho) && (\text{since $S$ is $(\epsilon/12$)-representative}) 
    \end{align*}
    Now, letting $x^\star \in \argmin_{x \in \budgets} \social(x)$, we observe:
    \begin{align*}
       \social(\overline{x}(S)) &\leq \panelcost(\overline{x}(S), S) + \epsilon/(4\rho) && (\text{by \Cref{claim}}) \\
        &\leq \rho \cdot \panelcost(x^\star, S)  + \tau + \epsilon/(4\rho) && (\text{by the assumption on $\overline{x}(S)$}) \\
        &\leq \rho \cdot \socialopt  + \tau + \epsilon/4 + \epsilon/(4\rho) && (\text{by \Cref{claim}})  \\
        &\leq \rho \cdot \socialopt + \tau+ \epsilon/2
    \end{align*}
Finally, we conclude:
    \begin{align*}
        &\mathbb{E}_{S \sim \uniformpanel} \left[ \social(\overline{x}(S)) \cdot \mathbbm{1} \left[ \text{$S$ is $(\epsilon/12)$-representative} \right] \right] \leq \rho \cdot \socialopt + \epsilon/2 + \tau, \\
        &\mathbb{E}_{S \sim \uniformpanel} \left[ \social(\overline{x}(S))  \cdot \mathbbm{1}\left[ \text{$S$ is not $(\epsilon/12)$-representative} \right] \right] \leq \epsilon/2 \\
        \Longrightarrow \quad &\mathbb{E}_{S \sim \uniformpanel} \left[ \social(\overline{x}(S)) \right] \leq \rho \cdot \socialopt + \epsilon + \tau
    \end{align*}
which completes the proof.
\end{proof}

Now, we establish an almost matching lower bound for the theorem above.

\begin{theorem}[Panel Complexity Lower Bound]\label{thm:pb_lower_bound}
Fix any number of projects $m \geq 2$, error $0 < \epsilon < 1/64$, and budget $B = \lfloor m/2 \rfloor$. 
Suppose there exists a panel decision function $\widetilde{x} : \mathcal{L}^k \to \mathcal{X}_B$ for some $k \geq 1$ such that,
for any participatory budgeting instance $\pbinstance$, the following holds:
    \begin{align*}
        \mathbb{E}_{S \sim \uniformpanel} \left[ \social (\widetilde{x}(S)) \right] \leq \socialopt + \epsilon.
    \end{align*}
Then it must be that $k = \Omega(m \cdot (1/\epsilon)^2)$.
\end{theorem}

\begin{proof}
   Our goal is to construct a function $g$ that satisfies the assumption of \Cref{lem:sample_lower_bound}, establishing a necessary lower bound on $k$.
    Let $h = \lfloor m / 2 \rfloor$ and $w = \lfloor (1/64) \cdot (1 / \epsilon) \rfloor$. 
    Note that $\epsilon \leq (1/64) \cdot (1/w)$
    {and that $1/w = \Omega(1/\epsilon)$}.
    Let $r > 1$ be {an integer parameter} large enough to satisfy the conditions of \Cref{lem:sample_lower_bound}. 
    
    For each $z \in \{-1,+1\}^{h}$, consider the camouflaged population of size $n = 2hwr$ with a feature $f(z, h, w, r)$ as defined in \Cref{def:almost_uniform}.
    We will now define an instance $\mathcal{I}(z) = \langle m, B, \cost_1^{(z)}, \ldots, \cost_n^{(z)} \rangle$ for every vector $z \in \{-1,+1\}^{h}$ by letting $\cost_i^{(z)}(x) = 1-x_j$ for each agent $i \in [n]$ such that $j = f_i(z,h,w,r)$.
    Note that exactly a $(1+z_j/w)/(2h)$-fraction of the population have cost function $1-x_{2j}$, while exactly a $(1-z_j/w)/(2h)$-fraction have cost function $1-x_{2j-1}$.
        Moreover, if $m$ is odd, then none of those cost functions depend on the $x_m$, i.e., the $m$-th project has no value to anyone.

  Next, we compare the expected social cost of any budget allocation with $\socialopt$.  
For every budget allocation $x \in [0,1]^m$ satisfying $\sum_{j=1}^{2\ell} x_j = B = h$, the social cost of $x$ in instance $\mathcal{I}(z)$ is given by:
    \begin{align*}
        \frac{1}{n} \sum_{i=1}^n \cost_i^{(z)}(x)  
        &= \sum_{j=1}^{h} \frac{(1+z_j/w) \cdot (1-x_{2j}) + (1-z_j/w) \cdot (1-x_{2j-1})}{2h} \\
        &= \sum_{j=1}^{h} \frac{(2 - x_{2j} - x_{2j-1}) - (z_j/w) \cdot (x_{2j} - x_{2j-1})}{2h} \\
        &= \sum_{j=1}^{h} \frac{2 - x_{2j} - x_{2j-1}}{2h} - \sum_{j=1}^{h} \frac{(z_j/w) \cdot (x_{2j} - x_{2j-1})}{2h} \\
        &= \frac{1}{2} - \sum_{j=1}^{h} \frac{(z_j/w) \cdot (x_{2j} - x_{2j-1})}{2h}.
    \end{align*}  

The optimal budget allocation for instance $\mathcal{I}(z)$ is to set $x_{2j-1} = 1$ and $x_{2j} = 0$ if $z_j = -1$, and $x_{2j-1} = 0$ and $x_{2j} = 1$ if $z_j = 1$. This allocation is budget-feasible since $B = h$. Thus,  
\begin{align*}
    \socialopt(\mathcal{I}(z)) = \frac{1}{n} \sum_{i=1}^n \cost_i^{(z)}(x) 
    = \frac{1}{2} - \sum_{j=1}^{\ell} \frac{1/w}{2h} 
    = \frac{1}{2} - \frac{1}{2w},
\end{align*}  
for every $z \in \{-1,+1\}^{h}$.  Henceforth, we will write $\socialopt = \socialopt(\mathcal{I}(z))$.

Since cost functions in the instances constructed above are monotone and they do not depend on $x_m$ if $m$ is odd, we may assume without loss of generality that  
$\sum_{j=1}^{2h} \widetilde{x}(c_1, \ldots, c_k) = B = h$ for every $c_1, \ldots, c_k \in \mathcal{L}$.  
By our assumption on $\widetilde{x} : \mathcal{L}^k \to \budgets$, for every $z \in \{-1,+1\}^{h}$, we have:  
\begin{align*}
    \mathbb{E}_{S \sim \uniformpanel} \left[ \frac{1}{2} - \sum_{j=1}^{\ell} \frac{(z_j/w) \cdot (\widetilde{x}_{2j}(S) - \widetilde{x}_{2j-1}(S))}{2h} \right] 
    &= \mathbb{E}_{S \sim \uniformpanel} \left[ \social (\widetilde{x}(S)) \right] \\
    &\leq \socialopt+ \epsilon \\
    &\leq \frac{1}{2} - \frac{1}{2 w} + \frac{1}{64 w} \\
    &\leq \frac{1}{2} - \frac{31}{64 w},
\end{align*}
since $\epsilon \leq (1/64) \cdot (1/w)$. Multiplying both sides by $2w$ and rearranging terms gives:  
\begin{align}
    \mathbb{E}_{S \sim \uniformpanel}\left[ \frac{1}{h} \sum_{j=1}^{h} z_j \cdot (\widetilde{x}_{2j}(S) - \widetilde{x}_{2j-1}(S)) \right] \geq \frac{31}{32}. \label{eq:previousbound}
\end{align}
From this, we argue that:
\begin{align}\label{eq:PCLB-social-opt-bound}
    \mathbb{P}_{S \sim \uniformpanel}\left[ \frac{1}{h} \sum_{j=1}^{h} \mathbbm{1}\left[ z_j \cdot (\widetilde{x}_{2j}(S) - \widetilde{x}_{2j-1}(S)) > 0 \right] \geq \frac{3}{4} \right] \geq \frac{7}{8} \quad \text{for all $z \in \{-1,+1\}^{h}$},
\end{align}
since if this probability were smaller than $7/8$, then we would get: 
\begin{align*}
     \mathbb{E}_{S \sim \uniformpanel}\left[ \frac{1}{h} \sum_{j=1}^{h} z_j \cdot (\widetilde{x}_{2j}(S) - \widetilde{x}_{2j-1}(S)) \right] 
     < \frac{7}{8} \cdot 1 + \frac{1}{8} \cdot \frac{3}{4} = \frac{31}{32},
\end{align*}
which would contradict the previous bound in Statement~(\ref{eq:previousbound}).

    Define $g : \mathcal{L}^k \to \{-1,+1\}^{h}$ as follows: 
\begin{align*}
    g_j(c) =
    \begin{cases}
        +1 & \text{if } \widetilde{x}_{2j}(c) - \widetilde{x}_{2j-1}(c) > 0, \\
        -1 & \text{otherwise}.
    \end{cases}
\end{align*}
Note that $g$ is independent of $z$.
We use $g(f_S(z,h,w,r))$ for $S \subseteq [n]$ to denote $g((\cost_i^{(z)})_{i \in S})$.

Using Statement (\ref{eq:PCLB-social-opt-bound}), we will argue that
$g(f_S(z,h,w,r))$ and $z$ agree on at least $(3/4) \cdot h$ indices with probability at least $7/8$,
which can be expressed as:
\[
\mathbb{P}_{S \sim \uniformpanel} \left[ \|g(f_S(z,h,w,r)) - z\|_1 \leq 1/4 \right] \geq 7/8 \quad \text{for all $z \in \{-1,+1\}^{h}$}.
\]
To see why this holds, note that:
\begin{align*}
    & z_j \cdot (\widetilde{x}_{2j}(S) - \widetilde{x}_{2j-1}(S)) > 0 \\
    \Longrightarrow \quad & (z_j = +1 \text{ and } \widetilde{x}_{2j}(S) - \widetilde{x}_{2j-1}(S) > 0) 
  \text{ or } (z_j = -1 \text{ and } \widetilde{x}_{2j}(S) - \widetilde{x}_{2j-1}(S) \leq 0) \\
    \Longleftrightarrow \quad & (z_j = +1 \text{ and }g_j(f_S(z,h,w,r)) = +1)   \text{ or } (z_j = -1 \text{ and } g_j(f_S(z,h,w,r)) = -1) \\
    \Longleftrightarrow \quad & g_j(f_S(z,h,w,r)) = z_j.
\end{align*}
Hence, the function $g$ satisfies the assumption of \Cref{lem:sample_lower_bound}, which concludes the proof.
\end{proof}

The final result shows that, even with an arbitrarily large panel, no panel decision function can guarantee a purely multiplicative approximation to the optimal allocation.

\begin{restatable}[Impossibility for Multiplicative Error Guarantee]{theorem}{pbimpossibilityefficiency}\label{thm:pb_impossibility}
Fix any number of projects $m \geq 2$, panel size $k \geq 1$, and budget $0 < B < m$. Then, for any panel decision function $\widetilde{x} : \mathcal{L}^k \to \mathcal{X}_B$,
 there exists a participatory budgeting instance $\pbinstance$  such that, for any $\rho > 1$, the following holds:
\begin{align*}
        \mathbb{E}_{S \sim \uniformpanel} \left[ \social(\widetilde{x}(S)) \right] > \rho \cdot \socialopt.
    \end{align*}
\end{restatable}

\subsection{Core Fairness in Participatory Budgeting}

We now turn to the question of the panel complexity of fairness in participatory budgeting.
In participatory budgeting, the \emph{core} is a group fairness property that enforces quite a strong requirement while also being guaranteed to exist under some natural conditions \citep{FainGM16}.
Informally, a budget allocation belongs to the core when no group of agents could, given their proportional share of the budget, select an alternative budget allocation under which each agent in the group incurs a strictly lower cost.
We define approximate core by introducing an additive error to the share of the budget given to a group of agents and both an additive and multiplicative error to the cost reduction requirement. 
\begin{definition}[Approximate Core Fairness for Participatory Budgeting]    \label{def:apx_core}
    Given a budget allocation $x \in \budgets$ and a subset of agents $S \subseteq [n]$, a subset $T \subseteq S$ is said to form an $(\eta, \tau, \rho)$-blocking coalition for some $\eta,\tau \geq 0$ and $\rho \geq 1$ if there exists an alternative budget allocation $x' \in \budgets$ such that:
    \begin{align*}
        \sum_{j=1}^m x'_j/B + \eta \leq |T|/|S| \quad\text{and}\quad\rho \cdot \cost_i(x') + \tau < \cost_i(x) \text{ for all }i \in T.
    \end{align*} 
A budget allocation $x \in [0,1]^m$ is in the $(\eta, \tau, \rho, S)$-core if no $(\eta, \tau, \rho)$-blocking coalition exists for $x$. 
The set of all budget allocations in the $(\eta,\tau,\rho, S)$-core for a panel $S$ is denoted by $\panelcore_S(\eta, \tau, \rho)$, while the set of all budget allocations in the $(\eta,\tau,\rho, [n])$-core for the entire population is denoted by $\socialcore(\eta,\tau,\rho)$.
\end{definition}

Hence, budget allocations in $\socialcore(0, 0, 1)$ are in the exact core.
Budget allocations in the exact core are known to exist under mild assumptions on the agents' cost functions \citep{Fole70a}. For instance, core existence is guaranteed under cost functions that have continuous derivatives and are strictly monotonic and strictly convex.
In any case, the main theorem of this section, which we state and prove next, states a panel complexity result for the approximate core that applies to all participatory budgeting settings in which the core can be approximated to a constant factor.

\begin{theorem}[Panel Complexity Upper Bound]  
\label{thm:core_ub}
Fix any instance $\pbinstance$ of the participatory budgeting problem.
    Let the panel size be:
    \begin{align*}
        k = \Omega\left( \left( \frac{1}{\epsilon}\right)^2 \left( m \log \left( \frac{1}{\epsilon} \right) + \log\left( \frac{1}{\delta}\right) \right) \right).
    \end{align*}
    Suppose that panel $S \subseteq [n]$ of size $k$ selects a budget allocation satisfying:
    \begin{align*}
        \overline{x}(S) \in \panelcore_S(\eta, \tau, \rho)
    \end{align*}
    for some $\eta, \tau \geq 0$ and $\rho \geq 1$, where $\rho$ is a constant.\footnote{If we were to extend the argument to non-constant $\rho$, the panel complexity would depend on $\rho$.} Then it holds that:
    \begin{align*}
        \mathbb{P}_{S \sim \uniformpanel}\left[ \overline{x}(S) \in \socialcore(\eta + \epsilon, \tau + \epsilon, \rho) \right] \geq 1-\delta.
    \end{align*}
\end{theorem}
\begin{proof}
   Consider an ${\epsilon}/{\min(4\rho,2B)}$-covering of $\budgets$; see \Cref{def:covering} for a formal definition of a covering of a metric space. 
   We have some $x_1, \ldots, x_K \in \budgets$ where $K = (1/\epsilon)^{O(m)}$ so that for every $x \in \mathcal{X}_B$ there is some $a \in [K]$ such that $\|x-x_a\|_1 \leq \epsilon/(4\rho)$ and $\|x-x_a\|_1 \leq \epsilon/(2B)$. 
   For every $a,b \in [K]$ and $i \in [n]$, we define a feature:
   \begin{align*}
        f_{a,b}(i) = \mathbbm{1}\left [\cost_i(x_a) > \rho\cdot \cost_i(x_b) + \tau + \epsilon/(4\rho) \right].
    \end{align*}
    Since there are $K^2 = (1/\epsilon)^{O(m)}$ features, it is by \Cref{thm:many_features} that, if the panel size is $k = O((1/\epsilon)^2 \cdot (\log(K^2) + \log(1/\delta))) = O((1/\epsilon)^2 \cdot (m \log (1/\epsilon) + \log(1/\delta)))$, then panel $S \subseteq [n]$ is $(\epsilon/2)$-representative of the features defined above with probability $1-\delta$.
    Formally, this means that the following holds for each $a,b\in[K]$ with probability at least $1-\delta$:
    \begin{align*}
       \left| \frac{1}{|S|} \sum_{i \in S} f_{a,b}(i) - \frac{1}{n} \sum_{i=1}^n f_{a,b}(i)\right| \leq \epsilon/2 .
    \end{align*}

    Due to this observation, we need only show that, for any $(\epsilon/2)$-representative panel $S \subseteq [n]$, it holds that $\overline{x}(S) \in \panelcore_S(\eta, \tau, \rho) \implies \overline{x}(S) \in \socialcore(\eta+\epsilon, \tau+\epsilon, \rho).$
    Let us fix an $(\epsilon/2)$-representative panel $S \subseteq [n]$.
    Suppose that an alternative $x \in [0,1]^m$ is not in the $(\eta+\epsilon, \tau+\epsilon, \rho)$-core for the whole population, i.e., $x\not \in \socialcore(\eta+\epsilon, \tau+\epsilon, \rho)$. 
    By the definition of approximate core, we know that there is some coalition $T \subseteq [n]$ and an alternative $y \in [0,1]^m$ such that:
    \begin{align*}
        \sum_{j=1}^m y_j/B + \eta+\epsilon \leq |T|/n \quad \text{ and } \quad \rho \cdot \cost_i(y) + \tau + \epsilon < \cost_i(x) \text{ for all $i \in T$}.
    \end{align*}

    Let $a \in [K]$ be such that $\|x - x_a\|_1 \leq \epsilon/(4\rho)$ and let $b \in [K]$ be such that $\|y - x_b\|_1 \leq \epsilon/(4\rho)$. 
    Let $y' = x_b$.
    We will argue that $x\not \in \panelcore_S(\eta,\tau,\rho)$ by showing that there exists a $(\eta,\tau,\rho)$-blocking coalition for $y'$.
    That is, we will show $\sum_{j=1}^m y'_j + \eta \leq |T'|/|S|$ where $T' = \{i \in S : \rho\cdot \cost_i(y') + \tau < \cost_i(x)\}.$

    Note that for every $i \in T$:
    \begin{align*}
        \rho \cdot \cost_i(x_b) + \tau 
        &\leq \rho \left[ \cost_i(y) + \epsilon/(4\rho)\right] + \tau && (\text{since $\|y-x_b\| \leq \epsilon/(4\rho)$}) \\
        &< \cost_i(x) - \epsilon + \epsilon/4 && (\text{by the core violation})\\
        &\leq \cost_i(x_a)-\epsilon + \epsilon/4 + \epsilon/({4\rho}) && (\text{since $\|x-x_a\| \leq \epsilon/(4\rho)$}) \\
        &\leq \cost_i(x_a) - \epsilon/({4\rho})
    \end{align*}
    which means that $ f_{a,b}(i) = 1$ for all $i \in T$ and consequently $(1/n)\sum_{i=1}^n f_{a,b}(i) \geq |T|/n$.

    At last, we see that:
    \begin{align*}
        \frac{|T'|}{|S|} &= \frac{1}{|S|}\sum_{i \in S} \mathbbm{1}[\rho \cdot \cost_i(y') + \tau < \cost_i(x)] && (\text{by the definition of $T'$})\\
        &\geq \frac{1}{|S|}\sum_{i \in S} \mathbbm{1}[\rho \cdot \cost_i(y') + \tau + {\epsilon}/({4\rho}) < \cost_i(x_a)] && (\text{since $\|x-x_a\| \leq {\epsilon}/({4\rho})$}) \\
        &= \frac{1}{|S|}\sum_{i \in S} f_{a,b}(i) && (\text{by the definition of $f_{a,b}(i)$}) \\\
        &\geq \frac{1}{n}\sum_{i=1}^n f_{a,b}(i) - \epsilon/2 && (\text{since $S$ is $(\epsilon/2)$-representative}) \\ 
        &\geq \frac{|T|}{n} - \epsilon/2 && (\text{since $f_{a,b}(i) = 1$ for all $i \in T$}) \\
        &\geq \sum_{j=1}^m y_j/B + \eta + \epsilon - \epsilon/2 && (\text{by the core violation}) \\
        &\geq \sum_{j=1}^m y'_j/B + \eta + \epsilon -\epsilon/2 - {\epsilon}/2 && (\text{since $\|y-y'\| \leq {\epsilon}/(2B)$}) \\
        &\geq  \sum_{j=1}^m y'_j/B + \eta
    \end{align*}
    as needed.
\end{proof}

\section{Facility Location via Sortition}\label{sec:facility}

In this section, we analyze the facility location problem. First, we define the model. Next, we establish tail bounds for the panel distribution. We then analyze the panel complexity of social welfare maximization before extending our model to the case of multiple facilities. The proofs omitted in this section are deferred to \Cref{apx:proofsfacilities}.

\subsection{Model of Facility Location}

An instance of facility location is a tuple $\facilityinstance$, where $(\mathcal{X}, d)$ denotes a metric space $\mathcal{X}$ equipped with a distance function $d : \mathcal{X} \times \mathcal{X} \to [0,1]$. 
A single facility must be placed at some point $q \in \mathcal{C} \subseteq \mathcal{X}$ and each agent $i \in [n]$ has an ideal location $x_i \in \mathcal{X}$, and incurs a cost equal to the distance from $q$, given by $d(q, x_i)$.  

\begin{definition}[Social Cost for Facility Location]\label{def:socialfacilitycost}
We define:
    \begin{enumerate}[noitemsep, topsep=0pt, label=(\roman*)]
    \item The {social cost} of $q \in \mathcal{C}$ as $\social(q) = (1/n) \cdot  \sum_{i=1}^n d(q, x_i)$,
    \item The {optimal social cost} as $\socialopt = \min_{q \in \mathcal{C}} \social(q)$,
    \item The {panel cost} of $q \in \mathcal{C}$ w.r.t. a panel $S$ as $\panelcost(q, S) = (1/k) \cdot \sum_{i \in S} d(q, x_i)$, and
    \item The {optimal panel cost} as $\panelopt(S) = \min_{q \in \mathcal{C}} \panelcost(q, S)$.
\end{enumerate}
\end{definition}

Throughout this section, we consider the optimal facility location $\overline{q}(S)$ with respect to $S$. {When we need to specify the facility location instance $I$, we write $\overline{q}(S,I)$.}
We also denote by $q^\star$ the optimal facility location for the entire population.  Formally:
\begin{align*}
    \overline{q}(S) \in \argmin_{q \in \mathcal{C}} \panelcost(q, S) \quad \text{and} \quad q^\star \in \argmin_{q \in \mathcal{C}} \social(q).
\end{align*}

\subsection{Bounding Panel Outliers in Facility Location}
\label{subsec:PanelOutliersRare}

A key question in evaluating the panel-optimal facility location $\overline{q}(S)$ for a panel $S$ drawn uniformly from the population is how far it deviates from the population-optimal facility $q^{\star}$. In this section, we analyze the behavior of $d(\overline{q}(S), q^\star)$ as a function of the panel size $k$.  

This analysis is useful for obtaining guarantees on social cost in \Cref{sec:facility_efficiency}, and offers an
alternative perspective on approximating the optimal facility location.
If $q^\star$ is unique, we naturally expect $d(\overline{q}(S), q^\star) \to 0$ as $k \to n$. Moreover, if this function decays smoothly with $k$, it offers insight into how well a panel of size $k$ approximates the optimal facility location.

To put this in more concrete terms, consider a government agency evaluating the feasibility of building a facility in a given region. While they may eventually be able to survey the entire population and determine the optimal location $q^\star$, logistical constraints such as land availability and permits may still pose challenges. 
The agency's primary concern is that, initially, they lack a clear understanding of where the facility should be placed.  
To gain confidence before making a large commitment, they first query a panel. This query returns a ball $B \subseteq \mathcal{X}$ centered at $\overline{q}(S)$, which, with some probability dependent on $k$, contains $q^\star$. If $B$ proves to be a viable location and the agency decides to proceed, no resources are wasted, as they can refine their analysis by surveying the remaining population $[n] \setminus S$. 

Our main result in this section is a panel complexity guarantee for general metric spaces. We provide a tail bound on the probability that $d(q^\star, \overline{q}(S))$ is large relative to $\socialopt$ in general metric spaces. Note that, since the distances $d$ are bounded by $1$, it follows that $\socialopt \leq 1$.

\begin{theorem}[Panel Complexity Upper Bound]\label{thm:tail_bound}
Let $\facilityinstance$ be a facility location instance.
    Fix any $T> 2$ and $\delta\in(0,1)$. Let the panel size be:
    \[
    k=\frac{2\cdot \log(1/\delta)}{\log\left({T^2}/{\left(4(T-1)\right)}\right)}.
    \]
    Then it holds that:
    \[
    \mathbb{P}_{S \sim \uniformpanel}[d(q^\star, \overline{q}(S)) \leq T \cdot \socialopt ] \geq 1-\delta.
    \]
\end{theorem}

We note that the denominator in the panel complexity is always 
positive, i.e., for any $T>2$, we have $\log\left({T^2}/{\left(4(T-1)\right)}\right) > 0$. Furthermore, as $T \to \infty$, the denominator grows like $\log(T)$.
{At the end of this section, we also show that \Cref{thm:tail_bound} is tight, in the sense that its guarantee does not hold for $T \leq 2$.}

Before proving \Cref{thm:tail_bound}, we note that it implies that, under an additional PAC-style guarantee with multiplicative error, we can obtain a multiplicative approximation of the optimal social cost.

\begin{restatable}[Multiplicative PAC Guarantee Suffices]{corollary}{MultiplicativePACGuarantee}\label{thm:reduction}
Let $\facilityinstance$ be a facility location instance. 
For any $0 < \epsilon < 1$, suppose the panel size satisfies $k \geq O(\log(1/\epsilon))$, and it holds that:
     \begin{align*}
    \mathbb{P}_{S \sim \uniformpanel} [\social(\overline{q}(S)) \leq (1 + \epsilon) \cdot \socialopt] \geq 1-\epsilon.   
     \end{align*}
Then it follows that:
    \begin{align*}
        \mathbb{E}_{S \sim \uniformpanel} [\social(\overline{q}(S))] \leq (1+ 5\epsilon ) \cdot \socialopt.
    \end{align*} 
\end{restatable}

The proof of \Cref{thm:reduction} follows readily by applying \Cref{thm:tail_bound} and the multiplicative PAC guarantee to 
$\mathbb{P}_{S \sim \uniformpanel}[\social(\overline{q}(S)) \geq x ]$.
We note that \citet*{DBLP:conf/ijcai/AnshelevichP16} show that even a panel consisting of a single agent, commonly referred to as random dictatorship, provides a $3$-approximation to the optimal social cost in expectation.
In contrast, \Cref{thm:reduction} can be applied to obtain a $(1 + \epsilon)$-approximation for any $\epsilon > 0$. In \Cref{sec:facility_efficiency},  we apply \Cref{thm:reduction} in the proof of our main result (\Cref{thm:assouad}).

\newcommand{\myscale}{3}
\newcommand{\rad}{0.75pt}
\newcommand{\wid}{0.03}
\newcommand{\topt}{0.7}
\newcommand{\tipscale}{0.75}

\def\pointsinside{{1.12,0.67},{0.76,1.15},{1,1.2},{0.87,0.8}}
\def\pointsoutsideup{{1.74,1.5}}
\def\pointsoutsidedown{{0.48,0.35}}

\newcommand{\arc}[2]{%
  \draw[dashed, thick, -{Stealth[scale=\tipscale]}] (#1,#2) arc[
    start angle={atan2(#2-1,#1-1)},
    end angle=0,
    radius={veclen(#1-1,#2-1)}
  ];
}
\def\splitandarc#1,#2\relax{\arc{#1}{#2}}

\newcommand{\arcandup}[2]{%
  \draw[dashed, thick] 
    (#1,#2) arc[
      start angle={atan2(#2-1,#1-1)},
      end angle=2,
      radius={veclen(#1-1,#2-1)}
    ];
  \draw[dashed, thick, -{Stealth[scale=\tipscale]}] 
    ({1 + veclen(#1-1,#2-1) * cos(2)}, 
     {1 + veclen(#1-1,#2-1) * sin(2)}) 
    -- ({1 + \topt}, {1 + veclen(#1-1,#2-1) * sin(2)});
}
\def\splitandarcandup#1,#2\relax{\arcandup{#1}{#2}}

\newcommand{\arcanddown}[2]{%
  \draw[dashed, thick] 
    (#1,#2) arc[
      start angle={atan2(#2-1,#1-1)},
      end angle=-2,
      radius={veclen(#1-1,#2-1)}
    ];
  \draw[dashed, thick, -{Stealth[scale=\tipscale]}] 
    ({1 + veclen(#1-1,#2-1) * cos(-2)}, 
     {1 + veclen(#1-1,#2-1) * sin(-2)}) 
    -- ({1 + \topt}, {1 + veclen(#1-1,#2-1) * sin(-2)});
}
\def\splitandarcanddown#1,#2\relax{\arcanddown{#1}{#2}}

\begin{figure}[t]
    \centering
    \begin{minipage}{0.49\textwidth}
        \centering
        \begin{tikzpicture}[scale=\myscale]
            \draw[->, thick] (-0.1,0) -- (2,0);
            \draw[->, thick] (0,-0.1) -- (0,2);

            \draw[-, thick] (1,1) -- (1+\topt, 1);
            \draw [decorate, decoration={brace, mirror, amplitude=6pt}, thick] 
    (1,1) -- (1+\topt,1) node [midway, below=6pt, scale=0.6] {$T \cdot \textsc{Social}\text{-}\textsc{Opt}$};

            \filldraw (1,1) circle (\rad);
            \draw[thick] (1,1) circle (\topt);
            \node at (0.9,1.02) {$q^{\star}$};

            \foreach \p in \pointsinside {
                \draw[thick] (\p) circle (\rad);
            }
            \foreach \p in \pointsoutsideup {
                \draw[thick] (\p) circle (\rad);
            }
            \foreach \p in \pointsoutsidedown {
                \draw[thick] (\p) circle (\rad);
            }
        \end{tikzpicture}
        \smallbreak
        (i) General Instance
    \end{minipage}%
    \hfill
    \hfill
    \begin{minipage}{0.49\textwidth}
        \centering
        \begin{tikzpicture}[scale=\myscale]
            \draw[->, thick] (-0.1,0) -- (2,0);
            \draw[->, thick] (0,-0.1) -- (0,2);

            \filldraw (1,1) circle (\rad);
            \draw[thick] (1,1) circle (\topt);
            \node at (0.9,1.02) {$q^{\star}$};

            \draw[thick, ->] (1,1) -- (2,1);

            \foreach \p in \pointsinside {
                \expandafter\splitandarc\p\relax
                \filldraw[white] (\p) circle (\rad);
                \draw[thick] (\p) circle (\rad);
            }

            \foreach \p in \pointsoutsideup {
                \expandafter\splitandarcandup\p\relax
                \filldraw[white] (\p) circle (\rad);
                \draw[thick] (\p) circle (\rad);
            }

             \foreach \p in \pointsoutsidedown {
                \expandafter\splitandarcanddown\p\relax
                \filldraw[white] (\p) circle (\rad);
                \draw[thick] (\p) circle (\rad);
            }
        \end{tikzpicture}
        \smallbreak
        (ii) Finite Interval Instance
    \end{minipage}

    \caption{Illustration of the reductions described in \Cref{lem:reduction_main,lem:further_line_reduction}, as used in the proof of \Cref{thm:tail_bound}. 
    In subfigure (i), the solid black circle represents the optimal facility location $q^\star$, while the empty circles indicate the agent locations in $\mathcal{X}$.
In subfigure (ii), the dashed lines show the mapping of agent locations from the original metric space $\mathcal{X}$ to their corresponding positions in the metric space $\mathcal{Y} = [0, T]$.
    }
    \label{fig:reduction}
\end{figure}
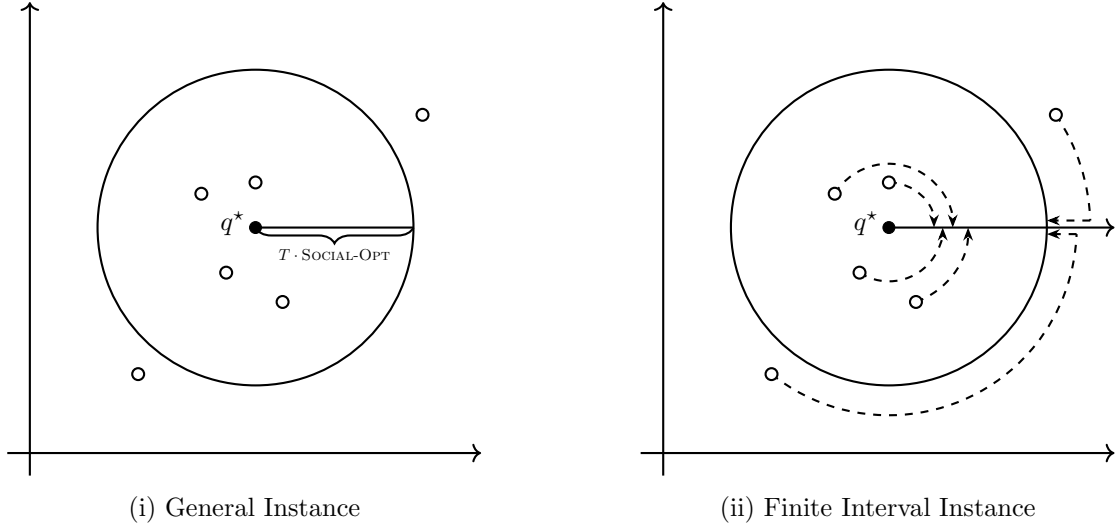

 In order to prove \Cref{thm:tail_bound},
 it is useful to reduce the problem to a metric space that is easier to analyze.
 The reduction must preserve certain properties on distances in order to yield an upper bound on $\mathbb{P}_{S \sim \uniformpanel}[d(q^\star, \overline{q}(S)) \leq T \cdot \socialopt ]$. 
 We find that a metric map centered at a specific point $p$, which we will now define, is sufficient for this purpose. 
 It is essentially a contraction that also preserves distances with respect to $p$, which, for our purposes, will correspond to the population optimum.

\begin{definition}[Metric Map]\label{def:metricmap}
    Let $(X,d_X)$ and $(Y,d_Y)$ be two metric spaces and $p \in X$. We call a map 
    $\pi: X \rightarrow Y$ a \emph{metric map} centered at $p$ if the following two conditions are satisfied:
    \begin{enumerate}[noitemsep, topsep=0pt,label=(\roman*)] 
    \item $d_X(x_1,x_2) \geq d_Y(\pi(x_1),\pi(x_2)) $ for all $ x_1,x_2 \in X$, and\label{list:metricmap1}
     \item $d_X(p,x) = d_Y(\pi(p),\pi(x)) $  for all $x \in X$. \label{list:metricmap2}
 \end{enumerate}
\end{definition}

{The following is a technical lemma which demonstrates that if an alternative $\bar{q}$ is far from the population optimum, then it must also incur significant social cost.}

\begin{lemma}\label{lem:alternative_far}
    Let $I = \langle \mathcal{X}, \mathcal{C}, d, (x_1,\ldots,x_n) \rangle$ be a facility location instance. Let $q_1, q_2 \in \mathcal{X}$ and $T \geq 0$ be such that $d(q_1, q_2) = T\cdot \social(q_1)$. Then, we have $\social(q_2) > (T-1)\cdot\social(q_1)$.
\end{lemma}

\begin{proof}
Denote by $S_1 = \{i \in [n] : \ d(q_1, q_2) \geq d(q_1,x_i) \}$, $S_2 = [n]\setminus S_2$ and $K = \social(q_1)$. We have:
\begin{align*}
    \sum_{i\in [n]} d(q_2,x_i) 
     & \geq \sum_{i\in S_1} \left(d(q_1,q_2) - d(q_1,x_i)\right) + 
    \sum_{i \in S_2}  \left(d(q_1,x_i) - d(q_1,q_2)\right) && (\text{by triangle inequality})\\
    & = \left(|S_1| - |S_2|\right)\cdot d(q_1, q_2)
     - \left(\sum_{i\in S_1} d(q_1,x_i) - \sum_{i \in S_2} d(q_1,x_i) \right)\\
    & = \left(|S_1| - |S_2|\right)\cdot d(q_1, q_2)
    - n\cdot K + 2\cdot \sum_{i\in S_2} d(q_1,x_i)\\
    & \geq \left(|S_1| - |S_2|\right)\cdot d(q_1, q_2)
    - n\cdot K + 2\cdot |S_2| \cdot d(q_1, q_2) && (\text{by definition of $S_2$})\\
    & = (|S_1|+|S_2|)\cdot d(q_1, q_2) - n\cdot K\\
    & = (T-1)\cdot n \cdot K && (\text{since } d(q_1, q_2) = T \cdot K).
\end{align*}
The result follows.
\end{proof}

{Using this lemma, we next show that if the set of non-optimal alternatives is constrained to contain only outliers, in the sense that $d(q^\star, q) \geq T \cdot \socialopt$ for any $q\neq q^*$, then finding a metric map centered at the population optimum suffices to define a valid reduction.}

\begin{lemma}[General Reduction]\label{lem:construct_I2}
Let $T > 2$ and $I = \langle X, \mathcal{C}_1, d_X, (x_1, \ldots, x_n) \rangle$ be a facility location instance with population optimum $q_1^\star \in \mathcal{C}_1$.
Suppose that for all $c \in \mathcal{C}_1 \setminus \{q_1^\star\}$, we have $d_X(q_1^\star, c) \geq T \cdot \socialopt(I)$. Let $(Y, d_Y)$ be a metric space, and let $\pi: X \rightarrow Y$ be a metric map centered at $q_1^\star$.
Then, there exists an instance $J = \langle Y, \mathcal{C}_2, d_Y, (y_1, \ldots, y_n) \rangle$ such that:
    \begin{enumerate}[noitemsep, topsep=0pt,label=(\roman*)] 
    \item for any $S \subseteq [n]$, if $d(\bar{q}(S, I), q^\star) \geq T \cdot \socialopt(I)$, then $d_Y(\bar{q}(S, J), q^\star) \geq T \cdot \socialopt(J)$,
    \item $\mathcal{C}_2 = \{\pi(q_1^\star)\} \cup \{y\in Y : d_Y(\pi(q_1^\star),y) \geq T\cdot \socialopt(J)\}$, and
    \item the population optimum in $J$ is $\pi(q_1^\star) \in \mathcal{C}_2$.
 \end{enumerate}
\end{lemma}
\begin{proof}
We construct $J$ by defining $\mathcal{C}_2 = \{\pi(q_1^\star)\} \cup  \{y\in Y : d_Y(\pi(q_1^\star),y) \geq T\cdot \socialopt(I)\} $ and $y_i = \pi(x_i)$ for all $i\in [n]$.
By Property~\ref{list:metricmap2} of \Cref{def:metricmap}, we have $\social(\pi(q_1^\star), J) = \social(q_1^\star, I) = \socialopt(I)$.
By the definition of $J$, it follows that $d_Y(\pi(q_1^\star), y) > 2 \cdot \socialopt(I)$ for all $y \in \mathcal{C}_2 \setminus \{\pi(q_1^\star)\}$.
Thus, by \Cref{lem:alternative_far}, we have $\social(y, J) > \social(\pi(q_1^\star), J)$ for all $y \in \mathcal{C}_2 \setminus \{\pi(q_1^\star)\}$, and therefore $\pi(q_1^\star)$ is the population optimum in $J$.
Consequently, $\socialopt(J) = \socialopt(I)$.

Let $S \subseteq [n]$ be such that $\bar{q}(S,I) \neq q_1^\star$. Then, we get:
\begin{align*}
    \sum_{i\in S} d_Y(\pi(q_1^\star),y_i) &= \sum_{i\in S} d_Y(\pi(q_1^\star),\pi(x_i)) && (\text{by definition of } y_i) \\
    &= \sum_{i\in S} d_X(q_1^\star,x_i) && (\text{by Property~\ref{list:metricmap2} of \Cref{def:metricmap}}) \\
    &> \sum_{i \in S}d_X(\bar{q}(S, I),x_i) && (\text{since } \bar{q}(S,I) \neq q_1^\star) \\
    &\geq \sum_{i\in S}d_Y(\pi(\bar{q}(S, I)), y_i) && (\text{by Property~\ref{list:metricmap1} of \Cref{def:metricmap}})
\end{align*}
implying $\bar{q}(S,J) \neq \pi(q_1^\star)$. This concludes the proof.
\end{proof}

Finally, we show how to make use of Lemma \ref{lem:construct_I2} by defining a specific metric map that projects an arbitrary metric space to $[0,\infty)$. 
\Cref{fig:reduction}(i) gives a visualization of this projection.

\begin{lemma}[Reduction to the Real Line]\label{lem:reduction_main}
Let $T > 2$ and $I = \langle X, \mathcal{C}_1, d_X, (x_1, \ldots, x_n) \rangle$ be a facility location instance with population optimum $q^\star$. Then, there exists an instance $J = \langle Y, \mathcal{C}_2, d_Y, (y_1, \ldots, y_n) \rangle$, where $Y = [0, \infty)$ and $d_Y(z_1, z_2) = |z_1 - z_2|$, such that:
    \begin{enumerate}[noitemsep, topsep=0pt,label=(\roman*)] 
    \item for any $S \subseteq [n]$, if $d(\bar{q}(S, I), q^\star) \geq T \cdot \socialopt(I)$, then $d_Y(\bar{q}(S, J), q^\star) \geq T \cdot \socialopt(J)$, 
    \item $\mathcal{C}_2 = \{0\} \cup [T \cdot \socialopt(J), \infty)$, and 
    \item the population optimum in $J$ is $0$.
 \end{enumerate}
\end{lemma}
\begin{proof}
Let $R = \langle X,\mathcal{C}_3, d_X, (x_1,\ldots,x_n) \rangle$ where 
$\mathcal{C}_3 = \{q^\star\} \cup \{x \in \mathcal{C} \ | \ d_X(q^\star,x) \geq T\cdot \socialopt(I) \}$. 
Take a subset $S \subseteq [n]$ such that $d_X(q^\star,\bar{q}(S,I)) \geq T \cdot \socialopt(I)$. Since $\mathcal{C}_1 \subseteq \mathcal{C}$, and $\bar{q}(S, I) \in \argmin_{q \in \mathcal{C}} \panelcost(q,S,I)$, we must have $\bar{q}(S, I_1) = \bar{q}(S, I)$. Now consider the metric space $Y = [0,\infty)$ with distance given by $d_Y(z_1,z_2) = |z_1-z_2|$. 
Further, define the map $\pi:\mathcal{X}\rightarrow Y$ as $\pi(x) = d_X(x,q^\star)$ for all $x\in \mathcal{X}$.
We verify that $\pi$ satisfies the conditions of a metric map centered at $q^\star$.
First, for any $x_1, x_2 \in \mathcal{X}$, the triangle inequality gives:
\[
d_Y(\pi(x_1),\pi(x_2)) = |d_X(q^\star,x_1) - d_X(q^\star, x_2)| \leq d_X(x_1,x_2),
\]
which establishes Property~\ref{list:metricmap1} of \Cref{def:metricmap}.
Next, for any $x \in \mathcal{X}$:
\[d_Y(\pi(q^\star),\pi(x)) = |d_X(q^\star,q^\star) - d_X(q^\star,x)| = d_X(q^\star,x).\]
verifying Property~\ref{list:metricmap2} of \Cref{def:metricmap}.
Therefore, $\pi$ is a valid metric map centered at $q^\star$. The desired result then follows directly from \Cref{lem:construct_I2}.
\end{proof}

{It is henceforth without loss of generality to assume that $\mathcal{X} = [0,\infty)$. 
But, in order to apply concentration bounds, we also cannot have points $x_i \rightarrow \infty$.}
Hence, we further reduce the problem by mapping all points $x_i$ that are larger than $T\cdot \opt$, to $T\cdot \opt$. 
This reduction to a finite interval instance is illustrated in \Cref{fig:reduction}(ii).

\begin{lemma}[Reduction to a Finite Interval]\label{lem:further_line_reduction}
    Let $T > 2$ and $I = \langle X, \mathcal{C}_1, d, (x_1,\ldots,x_n) \rangle$ be a facility location instance with $X = [0,\infty)$, $d(x,y) = |x - y|$, $\mathcal{C}_1 = \{0\} \cup [T \cdot \socialopt(I), \infty)$, and population optimum at $0$.
    Then, there exists $J = \langle Y, \mathcal{C}_2, d, (y_1,\ldots,y_n) \rangle$, where $Y = [0, T]$, such that:
        \begin{enumerate}[noitemsep, topsep=0pt,label=(\roman*)] 
    \item for any $S \subseteq [n]$, if $d(\bar{q}(S, I), q^\star) \geq T \cdot \socialopt(I)$, then $d_Y(\bar{q}(S, J), q^\star) \geq T \cdot \socialopt(J)$, 
    \item $\mathcal{C}_2 = \{0, T\}$, 
    \item the population optimum in $J$ is $0$, and
    \item $\socialopt(J) \leq 1$.
 \end{enumerate}
\end{lemma}

\begin{proof}
By scaling all $x_i$ by a factor of $\socialopt(I)$, we may assume that $\socialopt(I) = 1$.
Moreover, by relabeling the agents, we may assume that $x_1 \leq x_2 \leq \ldots \leq x_n$. Let $t$ be the smallest index such that $x_{t+1} \geq T$.
We define instance $J$ by setting $\mathcal{C}_2 = \{0, T\}$ and $y_i = \min(x_i, T)$.

First, we argue that the population optimum in $J$ is $0$. Since the population optimum in $I$ is $0$, we have:
\begin{align*}
    \sum_{i\in [n]} |y_i - T| = \sum_{i = 1}^n |T - x_i| - \sum_{i=t+1}^n |T - x_i| \geq \sum_{i=1}^n |x_i - 0| -  \sum_{i=t+1}^n |T - x_i| = \sum_{i=1}^n |y_i - 0|
\end{align*}
which proves the claim. It also follows that $\socialopt(J) = \sum_{i \in [n]} y_i \leq \sum_{i \in [n]} x_i = \socialopt(I) = 1$.

Let $S \subseteq [n]$ be a subset such that $\bar{q}(S, I) \neq 0$.
Define $q' = \bar{q}(S, I)$.
Consider the function $f_S : \mathbb{R}_{\geq 0} \rightarrow \mathbb{R}_{\geq 0}$ defined by $f_S(z) = \sum_{i\in S} |x_i - z|$. Clearly, $f_S(z)$ is unimodal and decreasing at $z = 0$.
By definition of $\overline{q}(S,I)$, we have that $\overline{q}(S,I) \in \argmin_{z\in \mathcal{C}_1}f_S(z)$.
If $q' > T$, then $q' = \argmin_{z\in \mathbb{R}_{\geq 0}}f_S(z)$. By unimodality, this implies $f_S(T) \leq f_S(0)$.
Moreover, observe that the reduction in social cost from $I$ to $J$ is the same for both $0$ and $T$, i.e., we have
 $\sum_{i\in S}x_i - \sum_{i\in S}y_i = \sum_{i\in S}|x_i - T| - \sum_{i\in S}|y_i - T|$. 
 Thus, if $\bar{q}(S, I) \neq 0$, then $T$ has lower social cost than $0$ in instance $I$, and therefore also in instance $J$.
Hence, if $\bar{q}(S, I) \neq 0$, it follows that $\bar{q}(S, J) = T$, and so $d_Y(\bar{q}(S, J), q^\star) \geq T \cdot \socialopt(J)$.
\end{proof}

Finally, we are ready to prove the main theorem of this section by a standard concentration argument. 
{We note that, in order to obtain a valid bound for any $T>2$, it is crucial to use the strongest version of Hoeffding's inequality, rather than the simplified version more commonly applied.}

\begin{proof}[Proof of \Cref{thm:tail_bound}]
By \Cref{lem:reduction_main,lem:further_line_reduction}, we only need to consider an instance $I = \langle {X},\mathcal{C}, d, (y_1,\ldots,y_n) \rangle$ where $X = [0,T]$, $d(x,y) = |x-y|$, $\mathcal{C} = \{0,T\}$, $\socialopt(I) \leq 1$, and the population optimum is $q^\star = 0$. 
Denote $\opt = \socialopt(I)$. Observe that for any panel $S \subseteq [n]$, we have:
\begin{align*}
   \overline{q}(S) = 0 \quad\Longleftrightarrow\quad \sum_{i \in S} x_i \leq \sum_{i \in S} (T-x_i) \quad\Longleftrightarrow\quad \sum_{i \in S} x_i \leq (T/2) \cdot k
\end{align*}
It follows that:
\begin{align}
\mathbb{P}_{S \sim \uniformpanel} \left[d(\overline{q}(S), q^\star) \geq T \right]
&= \mathbb{P}_{S \sim \uniformpanel} \left[ \sum_{i \in S} x_i > (T/2) \cdot k \right]\nonumber\\
&= \mathbb{P}_{S \sim \uniformpanel} \left[\frac{1}{k}\sum_{i\in S} \frac{x_i}{T} - \frac{\opt}{T} > \frac{1}{2} - \frac{\opt}{T} \right] \nonumber\\
&\leq \left(\left(\frac{\opt/T}{1/2}\right)^{1/2}\cdot \left(\frac{1-\opt/T}{1/2}\right)^{1/2}\right)^k \label{eq:hoeffding} \\
 &\leq \left(\left(\frac{1/T}{1/2}\right)^{1/2}\cdot \left(\frac{1-1/T}{1/2}\right)^{1/2}\right)^k \label{eq:quadratic} \\
&= \left(\frac{4(T-1)}{T^2}\right)^{k/2}. \nonumber
\end{align}
Inequality~(\ref{eq:hoeffding}) follows from the fact that $x_i/T \in [0,1]$ and $(1/n) \cdot \sum_{i \in [n]} x_i/T = \opt/T$, allowing us to apply Hoeffding's inequality \citep[Equation 2.1]{Hoeffding63}, which holds for sampling without replacement.
Inequality~(\ref{eq:quadratic}) holds because the function $x(1 - x)$ is increasing on the interval $[0, 1/2]$, and thus $(\opt/T) \cdot (1 - \opt/T) \leq (1/T) \cdot (1 - 1/T)$, given that $\opt \leq 1$ and $T > 2$.
Finally, the result follows by choosing $k$ such that:
\[\delta =\left(\frac{4(T-1)}{T^2}\right)^{k/2}
\quad\Longleftrightarrow\quad k=\frac{2\cdot \log(1/\delta)}{\log({T^2}/{(4(T-1)}))}\]
which concludes the proof.
\end{proof}

We now show that the guarantee of \Cref{thm:tail_bound} is tight, in the sense that it does not hold for $T \leq 2$.

\begin{restatable}[Panel Complexity Lower Bound]{theorem}{StarLowerBound}\label{prop:star_lower_bound}
    For any $k > 0$, there exists a facility location instance $\facilityinstance$ 
    such that for any panel decision function $\widetilde{q} : \mathcal{X}^k \to \mathcal{C}$, it holds that:
    \[
    \mathbb{P}_{S \sim \uniformpanel}[d(\widetilde{q}(S),q^{\star}) \geq 2 \cdot \socialopt] \geq 1/4.
    \]
\end{restatable}

\subsection{Welfare Maximization in Facility Location}    \label{sec:facility_efficiency}

This section examines the social cost minimization problem in facility location. To present our main result, we first introduce the concept of Assouad dimension, which provides an upper bound on the covering number of balls at varying scales.
For a formal definition of a covering in a metric space, see \Cref{def:covering}.

\begin{definition}[Assouad dimension]\label{def:assouad}
    We say that a metric space $(\mathcal{X}, d)$ has Assouad dimension $t$ if for every $r < R$, we can cover any ball $B(x, R)$ for $x \in \mathcal{X}$ with $C \cdot (R/r)^{t}$ balls of radius $r$ for some constant $C$.
\end{definition}

We are able to give panel complexity result for social cost minimization in facility location depending on the Assouad dimension of the metric space. 
Notably, any bounded subset of $\mathbb{R}^t$ has Assouad dimension $t$ under any $\ell_p$ norm.

\begin{restatable}[Panel Complexity Upper Bound]{theorem}{AssouadUpperBound}\label{thm:assouad}
   Let $\facilityinstance$ be a facility location instance, where $\mathcal{X}$ has Assouad dimension $t$. Let the panel size be:
    \begin{align*}
          k = O\left( \left( \frac{1}{\epsilon}\right)^2 \left( t \log \left( \frac{1}{\epsilon} \right) \right) \right).
    \end{align*}
    Then it holds that:
    \begin{align*}
        \mathbb{E}_{S \sim \uniformpanel} \left[ \social(\overline{q}(S)) \right] \leq (1+\epsilon) \cdot \socialopt.
    \end{align*}
\end{restatable}
\begin{proof}
Throughout this proof, we denote $\opt := \socialopt$.
Let $T$ be a sufficiently large constant to ensure that \Cref{thm:tail_bound} applies.
   Let $q^\star$ be the optimal facility location. Since $\mathcal{X}$ has Assouad dimension $t$, there is some sequence $q_1, \ldots, q_K \in B(q^\star, T \cdot \opt)$ with $q_1 = q^\star$ such that $B(q^\star, T \cdot \opt) \subseteq \bigcup_{j=1}^K B(q_j, (\epsilon/3) \cdot \opt)$ where $K = O((1/\epsilon)^t)$.

    For $1 \leq j \leq K$, let
    $X_j = \social(q_j, S) - \opt$ and $Y_j = \panelcost(q_j,S) - \panelcost(q^\star, S) + (\epsilon/3) \cdot \opt$.
    Let us also define the following events:
    \begin{align*}
        A &= \{S \subseteq [n] \,:\, \overline{q}(S) \in B(q^\star, T \cdot \opt) \} \\
        B_j &= \{ S \subseteq [n] \,:\,  X_j  \leq Y_j \} 
    \end{align*}
    for $1 \leq j \leq K$.
    We will argue that if $A \cap B_1 \cap \ldots \cap B_k$  holds, then $\social(\overline{q}(S)) \leq (1+\epsilon) \cdot \opt$ and that $\mathbb{P}[A \cap B_1 \cap \ldots \cap B_K] \leq \epsilon$. From this the result follows by  \Cref{thm:reduction}.

    First, observe that $\mathbb{P}_{S \sim \uniformpanel}[A] \geq 1-\epsilon/2$ by \Cref{thm:tail_bound}.
    For each $j \in [K]$ and $i \in [n]$, define a feature $f_j(i) = (1/\opt) \cdot (d(x_i,q_j) - d(x_i,q^\star))$. Note that:
    \begin{align*}
        (1/|S|) \cdot \sum_{i\in S} f_j(i) &= (1/\opt) \cdot (\panelcost(q_j,S) - \panelcost(q^\star, S)) \\
                (1/n) \cdot \sum_{i \in [n]} f_j(i) &= (1/\opt) \cdot (\social(q_j, S) - \opt).
    \end{align*}
    
    If $S$ is $(\epsilon/3)$-representative with respect to feature $f_j$, then $S \subseteq B_j$ because  $(1/|S|) \cdot \sum_{i \in S} f_j(i) \leq (1/n) \cdot \sum_{i \in [n]} f_j(i) + \epsilon/3$. Multiplying both sides by $\opt$ establishes the claim.  
   Moreover, since \( d(q^\star, q_i) \leq T \cdot \opt \), it follows that  
 $|f_j(i)| \leq (1/\opt) \cdot (T \cdot \opt) = T$.
    By \Cref{thm:many_features}, the panel $S$ is therefore $(\epsilon/3)$-representative with respect to all $K$ features simultaneously, implying that  $\mathbb{P}_{S \sim \uniformpanel}[B_1 \cap \ldots \cap B_k] \geq 1-\epsilon/2$.

    Finally, consider any panel $S \in A \cap B_1 \cap \ldots \cap B_k$. 
    Let $j$ be such that $\overline{q}(S) \in B(q_j, (\epsilon/3) \cdot \opt)$. Observe that the following derivation holds:
    \begin{align*}
        &\social(\overline{q}(S)) - \opt \\
        &\leq \social(q_j) - \opt + (1/3) \cdot \epsilon \cdot \opt  \\
        &\leq \panelcost(q_j,S) - \panelcost(q^\star,S) + (2/3) \cdot \epsilon \cdot \opt  \\
        &\leq \panelcost(\overline{q}(S),S) - \panelcost(q^\star,S) +  \epsilon \cdot \opt  \\
        &= \panelopt - \panelcost(q^\star,S) +  \epsilon \cdot \opt \\
        &\leq 0 + \epsilon \cdot \opt 
    \end{align*}
   The first and third inequalities follow from $d(\overline{q}(S), q_j) \leq (\epsilon/3) \cdot \opt$, the second inequality follows from $S \in B_j$, and the fourth inequality follows from the choice of $\overline{q}(S)$.  
    This concludes the proof.
\end{proof}

We also establish a nearly matching lower bound for the metric space $([0,1]^t, \ell_{\infty})$, which has an Assouad dimension of $t$.

\begin{restatable}[Panel Complexity Lower Bound]{theorem}{AssouadLowerBound}\label{thm:assouad_lower}
   Fix $\mathcal{X} = [0,1]^t$ with $d = \ell_{\infty}$. Let $\mathcal{C} = [1/4, 3/4]^t$ and $0 < \epsilon < 1/64$.
    Suppose there exists a panel decision function 
   $\widetilde{q} : \mathcal{X}^k \to \mathcal{C}$ for some $k \geq 1$
   such that, for any facility location instance
   $\facilityinstance$, the following holds:
    \begin{align*}
        \mathbb{E}_{S \sim \uniformpanel} \left[ \social(\widetilde{q}(S)))  \right] \leq (1+\epsilon) \cdot \socialopt.
    \end{align*}
    Then it must be that $k = \Omega((1/\epsilon)^2 \cdot t)$.
\end{restatable}

\subsection{Extension to Multiple Facilities}\label{sec:multifacilities}

In this section, we consider a setting where we can place facilities at multiple locations. An instance in this case is defined by $\committeeinstance$, where, as before, $(\mathcal{X}, d)$ is a metric space, $\mathcal{C} \subseteq \mathcal{X}$ is the set of possible facility locations, and $(x_1, \ldots, x_n)$ are the ideal locations of the agents in $[n]$. The additional parameter $\ell \geq 1$ specifies the desired number of facilities.  

Given a set $y = (y_1, \ldots, y_{\ell})$ of $\ell$ facility locations, it is natural to define the cost of agent $i$ as the distance to the closest facility. That is, we define $\cost_i(y) = \min_{j \in [\ell]} d(x_i, y_j)$. The social and panel costs, $\social(q)$ and $\panelcost(q, S)$, along with their respective optimal values, $\socialopt$ and $\panelopt$, are defined analogously to \Cref{def:socialfacilitycost}. 
We use $\overline{y}(S) \in \mathcal{C}^{\ell}$ and $y^* \in \mathcal{C}^{\ell}$ to denote the panel-optimal and social-optimal collection of facility locations, respectively.

Our first result applies to general metric spaces and shows that if we treat the locations of individuals in the population as a feature, the expected social cost can be bounded as follows.

\begin{restatable}{lemma}{LemmaMultiFacilities}\label{lem:multifacilities}
Let $\committeeinstance$ be a multiple facility location instance.  
Define $f(i) = x_i$ as the feature representing each agent's location.  
Let $\phi_f^{[n]}$ denote the population distribution of this feature and $\phi_f^S$ the panel distribution. Then, we have:
  \begin{align*}
        \mathbb{E}_{S \sim \uniformpanel} \left[ \social(\overline{y}(S)) \right] \leq \socialopt + \mathbb{E}_{S \sim \uniformpanel} \left[ W(\phi_f^{[n]}, \phi^S_f) \right]
    \end{align*}    
\end{restatable}
\begin{proof}
Let $S$ be any panel of size $k$ and $\overline{y}(S)$ the set of facilities minimizing the social cost for $S$. In the following calculations, we will consider $S$ fixed and write $\overline{y}$ in place of $\overline{y}(S)$. Let us denote, as usual, with $\delta_x$ the Dirac mass at point $x \in \mathcal{X}$. Our first remark is that, for any probability distribution $\lambda = \sum_{j \in [\ell]} \lambda_j \delta_{y_j}$ supported on the $\ell$ positions of the facilities in $\overline{y}$, the following holds: $\social(\overline{y}) \leq W(\phi_f^{[n]}, \lambda)$.

Indeed, $\cost_i(\overline{y})$ is equal to $\min_{j\in[\ell]} d(x_i, \overline{y}_j)$, and, as a consequence, $\cost_i(\overline{y}) \le \sum_j \lambda_j \cdot d(x_i, y_j)$ since this is a convex combination of positive values. Moreover, if $\gamma$ is a coupling of $\phi, \lambda$, by definition we have that: $1/n = \sum_j \gamma(x_i, y_j)$. Since this holds for every coupling $\gamma$, it holds for the minimal one as well, implying the above inequality.
Let us define the following function:
$\overline{\alpha}: [n] \times S \rightarrow \mathcal{X}^2$ such that $\overline{\alpha}(i,j) = (x_i, \alpha(z_j))$, where $\alpha(z_j)$ is a facility that realizes the minimum distance from $z_j$, namely $d(z_j, \alpha(z_j)) = \cost_{z_j}(\overline{y})$ (we wrote $z_j$ in place of the corresponding agent, to avoid confusion with the indexes). 
Now, call $\alpha(\phi^S_f)$ the distribution defined by pushing forward the uniform distribution over the $k$ members of panel $S$ through $\alpha$. If $\gamma$ is the coupling realizing the minimum in the definition of $W(\phi_f^{[n]}, \phi_f^S)$, then the pushforward measure $\overline{\alpha}(\gamma)$ is a valid coupling between $\phi$ and  $\alpha(\phi_f^S)$ and the following holds:
\begin{align*}
    \social(\overline{y}) &\le W\left(\phi_f^{[n]}, \alpha(\phi_f^S)\right)\\
    &\le \sum_{i\in[n]}\sum_{j\in[\ell]} d(x_i, \overline{y}_j) \cdot \overline{\alpha}(\gamma)(x_i, \overline{y}_j)\\
    &= \sum_{i\in[n]}\sum_{j\in[k]} d(x_i, \alpha(z_j)) \cdot \gamma(x_i, z_j)\\
    &\le \sum_{i\in[n]}\sum_{j\in[k]} d(x_i, z_j) \cdot \gamma(x_i, z_j) + \sum_{i\in[n]}\sum_{j\in[k]} d(z_j, \alpha(z_j)) \cdot \gamma(x_i, z_j)\\
    &= W(\phi_f^{[n]}, \phi_f^S) + \sum_{i\in[n]}\sum_{j\in[k]} d(z_j, \alpha(z_j)) \cdot \gamma(x_i, z_j)\\
    &= W(\phi_f^{[n]}, \phi_f^S) + \frac{1}{k} \sum_{j\in[k]} d(z_j, \alpha(z_j))\\
    &= W(\phi_f^{[n]}, \phi_f^S) + \panelopt(S),
\end{align*}
where $\panelopt(S)=(1/k) \cdot \sum_{j \in S} \cost_j(\overline{y})$.
Now, by taking the expectation of the quantities above over the choice of the panel $S$, we get:
\begin{equation*}
    \mathbb{E}_{S \sim \uniformpanel} \left[ \social(\overline{y}(S)) \right] \leq \mathbb{E}_{S \sim \uniformpanel}\left[W(\phi_f^{[n]}, \phi_f^S)\right] + \mathbb{E}_{S \sim \uniformpanel}\left[\panelopt(S)\right]
\end{equation*}
It remains to prove that the expected value of $\panelopt(S)$ is less than $\socialopt$. Indeed, by comparing the panel cost of $\overline{y}$ with that of the social optimum $y^*$, we can write the following:
\begin{align*}
   \mathbb{E}_{S \sim \uniformpanel} [\panelopt(S)] &\le  \mathbb{E}_{S \sim \uniformpanel} \left[\panelcost(y^\star, S) \right]\\ 
   &= \frac{1}{k} \sum_{i \in [n]} \left({\binom{n-1}{k-1}} / {\binom{n}{k}} \right) \cdot \cost_i(y^\star) = \socialopt,
\end{align*}
which concludes the proof.
\end{proof}

This lemma is closely tied to our notion of representativeness. In particular, if a panel is $\epsilon$-representative, it yields an additive bound on the expected social cost. The next result leverages Lemma~\ref{lem:multifacilities} to upper bound the panel complexity required to obtain such an additive bound in the special case where $\mathcal{X} = [0,1]$.
\begin{restatable}[Panel Complexity Upper Bound]{theorem}{PanelComplexityMultifacilities}\label{thm:line}
    Let $\committeeinstance$ be a multiple facility location instance with $\mathcal{X}=[0,1]$. Let the panel size be $k = O(1/\epsilon^{2})$. Then it holds that:
    \begin{align*}
        \mathbb{E}_{S \sim \uniformpanel} \left[ \social(\overline{y}(S)) \right] \leq \socialopt + \epsilon.
    \end{align*}
\end{restatable}
The proof of the result above relies on Lemma~\ref{lem:W1expectation} and on known results regarding the Wasserstein metric and is deferred to Appendix~\ref{apx:multifacilities}.
Finally, we show that in the multiple facility location problem, achieving a multiplicative error guarantee is impossible.  
\begin{restatable}[Impossibility for Multiplicative Error Guarantee]{theorem}{ImpossibilityMultiFacility}
There exists a metric space $\mathcal{X}$, a set of alternatives $\mathcal{C}$, a distance function $d$, and a number of facilities $\ell$ such that for any $k < n$ and any panel decision function $\widetilde{q} : \mathcal{X}^k \to \mathcal{C}^{\ell}$, there exist locations $x_1,\ldots, x_n$ such that, for any $\rho > 1$:
\begin{align*}
    \mathbb{E}_{S \sim \uniformpanel}\left[ \social(\widetilde{q}(S)) \right] > \rho \cdot \socialopt.
\end{align*}
\end{restatable}
\begin{proof}
Consider $\mathcal{X} = [0,1]$, $\mathcal{C} = \{0, 1/2, 1\}$, and $\ell = 2$.  
Fix any panel decision function $\widetilde{q} : \mathcal{X}^k \to \mathcal{C}^{\ell}$.  
Let $(q_1, q_2) = \widetilde{q}(0, 0, \dots, 0)$.  
Without loss of generality, assume that $q_1 \neq 1$ and $q_2 \neq 1$.  

Consider a population of agents located at $x_1 = \dots = x_{n-1} = 0$ and $x_n = 1$.  
Note that $\opt = 0$ since placing one facility at $0$ and another at $1$ achieves zero cost.  
Since the panel decision function selects $(0, 1/2)$ with nonzero probability, the result follows.  
\end{proof}

\section{Conclusion and Open Problems}\label{sec:conclusion}

In this work, we developed a framework for analyzing panel complexity and provided efficiency and fairness guarantees for two social choice problems: participatory budgeting and facility location. Our results open several interesting directions for future research.

First, throughout this work, we focused on the case where the panel is selected uniformly at random from the population. 
However, in many practical settings, stratified sampling is necessary. 
Extending the panel complexity analysis to non-uniform lotteries over panels remains an open problem.
More fundamentally, our notion of representativeness based on the Wasserstein metric may prove a useful way of framing representativeness in contexts transcending uniform sampling.

Second, we applied our framework to participatory budgeting and facility location, but many other important social choice problems remain to be studied, such as indivisible participatory budgeting \cite{rey2023computational}.
\citet*{MiSh20} give a result in the setting of centroid clustering which can be interpreted as a panel complexity of fairness in the multiple facilities extension we study in \Cref{sec:multifacilities}. 
Their notion of fairness, proportional fairness, was first defined by \citet*{CFLM19} and is inspired by the notion of core.
The result of \citet*{MiSh20} applies  to Euclidean space with $\ell_2$ distances, 
and uses a VC-dimension bound. Since this technique does not easily extend to other metrics, the panel complexity of fairness in multiple facility location on general metric spaces is an interesting open problem.

Third, we point out that our bounds do not depend on the population size, $n$. 
An open question is whether they can be improved when the population size is bounded.
Finally, we do not explicitly define the 
{panel decision functions,}
assuming only that they satisfy certain efficiency or fairness guarantees. 
Understanding how specific panel decision functions,
such as concrete voting rules for participatory budgeting or facility location,
impact panel complexity is an interesting direction for future work. 
We note, however, that our lower bounds on panel complexity hold for any 
{panel decision function}.

\bibliographystyle{abbrvnat}
\bibliography{references}

\appendix
 \section{Metric Space Preliminaries}\label{sec:metricpreliminaries}

The following are basic definitions of metric spaces used throughout this work.

\begin{definition}[Metric Space]
A pair $(\mathcal{X}, d)$, where $d : \mathcal{X} \times \mathcal{X} \to [0,1]$, is a metric space if:
\begin{enumerate}
    \item $d(x,x) = 0$ for all $x \in \mathcal{X}$ (identity),
    \item $d(x,y) = d(y,x)$ for all $x,y \in \mathcal{X}$ (symmetry),
    \item $d(x,z) \leq d(x,y) + d(y,z)$ for all $x,y,z \in \mathcal{X}$ (triangle inequality).
\end{enumerate}
The $r$-ball centered at $x \in \mathcal{X}$ is defined as $B_r(x) = \{ y \in \mathcal{X} \mid d(x,y) \leq r \}$.
\end{definition}

We also define a covering of a metric space.

\begin{definition}[$\epsilon$-Covering]\label{def:covering}
A set $(x_1, \dots, x_k) \in \mathcal{X}^k$ is an $\epsilon$-covering of $(\mathcal{X}, d)$ if the union of $\epsilon$-balls covers the entire space, i.e., $\mathcal{X} \subseteq \bigcup_{i=1}^k B_\epsilon(x_i)$. The covering number $N(\mathcal{X}, d, \epsilon)$ is the minimal size of an $\epsilon$-covering of $(\mathcal{X}, d)$.
\end{definition}

\begin{remark}[Covering Number of Euclidean Space]\label{rem:euclidean_covering}
For every $p \geq 0$ and $m \geq 1$, the Euclidean space $([0,1]^m, \|\cdot\|_p)$ is defined as $([0,1]^m, d)$, where $d(x,y) = \|x - y\|_p$.
The covering number of the Euclidean space satisfies $N([0,1]^m, \|\cdot\|_p, \epsilon) = O((1/\epsilon)^m)$.
\end{remark}

\section{Missing Remarks of Section 2}\label{sec:examplesthree}

\begin{remark}[Minimzing Sum of Distances is Not Representative]\label{rem:sumdist}
\citet*{DBLP:journals/corr/abs-2406-00913} consider a model in which agents are embedded in a similarity metric space.
They argue that, intuitively, selecting a panel to minimize the sum of distances from each agent to their nearest panel member does not necessarily yield a representative panel \cite[Example 7]{DBLP:journals/corr/abs-2406-00913}.
To illustrate this, they construct an example with two agents located at extreme positions $A$ and $B$, far from all other agents and from each other. The remaining $n - 2$ agents are evenly split between two locations, $C$ and $D$, which are 10 units apart.

For a panel size of $k = 3$, this distance-minimizing method would select the two outliers at $A$ and $B$, along with only one agent from the majority group, resulting in a highly non-representative panel. This is problematic because most of the population resides in either $C$ or $D$.

Using our notation, this example corresponds to a feature vector
$f = (-\infty, -5, \ldots, -5, 5, \ldots, 5, \infty)$,
where the number of $-5$ values is $(n - 1)/2$ and the number of $5$ values is also $(n - 1)/2$.
To satisfy our definition of $\epsilon$-representativeness (\Cref{def:representativeness}) for any $\epsilon > 0$, the panel would need to include the entire population, i.e., it must be that $k = n$.
Thus, our definition aligns with the intuition of \citet*{DBLP:journals/corr/abs-2406-00913}.
\end{remark}

\begin{remark}[Comparing Sampling With and Without Replacement]\label{remark:withwithout}
Intuitively, sampling without replacement yields a more representative panel than sampling with replacement.

Consider first the case where the population size $n$ is large. By the union bound, the probability of selecting the same individual more than once when sampling with replacement is at most $\binom{k}{2}/n$. Thus, as $n$ grows large, the two sampling schemes become nearly equivalent.

In the other extreme case of a small population with $n = 2$, sampling without replacement guarantees that a panel of size $k = 2$ includes both agents. In contrast, sampling with replacement requires a panel of size at least $\log(1/\delta) + 1$ to ensure, with probability at least $1 - \delta$, that both individuals are included. This illustrates that sampling with replacement can lead to higher panel complexity to achieve certain desirable guarantees.

\Cref{lem:W1expectation} extends this intuition beyond these extreme cases.
\end{remark}

\begin{remark}[Counterexample to Stochastic Dominance]\label{ex:sd}
    Using the notation of \Cref{lem:W1expectation}, we show that there exists a population of $n$ agents and a feature $f : [n] \to [0,1]$ such that:
    \begin{align}
        \mathbb{P}_{S \sim \mathcal{R}_{k,n}} \left[ W(\phi_f^{[n]}, \phi_f^{S}) \leq \delta \right] 
        > 
        \mathbb{P}_{S \sim \mathcal{U}_{k,n}} \left[ W(\phi_f^{[n]}, \phi_f^{S}) \leq \delta \right] 
        \quad \text{for some } \delta > 0,
        \label{eq:stochdominance}
    \end{align}
    meaning that the (first-order) stochastic dominance condition may fail.

    Consider a population of size $n = 5$ with feature values given by the vector $f = (0, 1/2, 1/2, 1/2, 1)$. Let the panel size be $k = 2$. We will demonstrate that:
    \begin{align}
        \mathbb{P}_{S \sim \mathcal{R}_{k,n}} \left[ W(\phi_f^{[n]}, \phi_f^{S}) \leq 0.2 \right] 
        \geq 
        9/25 = 0.36,
        \quad\text{while}\quad 
        \mathbb{P}_{S \sim \mathcal{U}_{k,n}} \left[ W(\phi_f^{[n]}, \phi_f^{S}) \leq 0.2 \right] 
        = 
        3/10 = 0.3.
        \label{eq:exam}
    \end{align}
    This contradicts the inequality in \Cref{eq:stochdominance}.
    We first compute the Wasserstein distance for different panels based on the feature values of their members:
    \begin{align*}
        \text{If } f(S) &= \{0, 1/2\}, \text{ then } 
        W(\phi_f^{S}, \phi_f^{[n]}) = (1/2) \cdot \left( 1/2 - 1/5 \right) + (1/2) \cdot (1/5) = 1/4 = 0.25. \\
        \text{If } f(S) &= \{1/2, 1/2\}, \text{ then } 
        W(\phi_f^{S}, \phi_f^{[n]}) = (1/2) \cdot (1/5) + (1/2) \cdot (1/5) = 1/5 = 0.2. \\
        \text{If } f(S) &= \{0, 1\}, \text{ then } 
        W(\phi_f^{S}, \phi_f^{[n]}) = (1/2) \cdot \left( 1/2 - 1/5 \right) + (1/2) \cdot \left( 1/2 - 1/5 \right) = 0.3.
    \end{align*}
    The case $f(S) = \{1/2, 1\}$ is symmetric to $f(S) = \{0, 1/2\}$.

    Note that $W(\phi_f^{S}, \phi_f^{[n]}) \leq 0.2$ holds only when $f(S) = \{1/2, 1/2\}$. Thus, \Cref{eq:exam} follows since the probability of drawing two agents with feature value $1/2$ when sampling with replacement is $(3/5)^2 = 9/25$, and when sampling without replacement it is $\binom{3}{2} / \binom{5}{2} = 3/10$.

    Note that the calculations above yield only a lower bound on 
    $\mathbb{P}_{S \sim \mathcal{R}_{k,n}} [ W(\phi_f^{[n]}, \phi_f^{S}) \leq 0.2 ]$, 
    since we did not consider panels with $f(S) = \{0,0\}$ or $f(S) = \{1,1\}$. However, these combinations occur with zero probability when sampling without replacement.
\end{remark}

\section{Missing Proofs of Section 2}\label{sec:missingthree}

\subsection{Proof of Lemma 2.3}
The following lemma follows from the proof of \citet*[Theorem 4]{Hoeffding63}. Although \citet*{Hoeffding63} established the result only for the case $V = \mathbb{R}$, it has been noted by \citet*{polaczyk2023concentrationboundssamplingreplacement} and \citet*{gross2010notesamplingreplacingfinite} that the proof extends to any normed vector space without modification.

\begin{lemma}[\citet*{Hoeffding63}]\label{lem:hoeffdingexpectation}
    Let $V$ be a normed vector space,  $x_1, \ldots, x_n \in V$, and $\psi : V \to \mathbb{R} \cup \{+\infty\}$ be a convex function. Then, it holds that:
    \begin{align*}
    \mathbb{E}_{S \sim \uniformpanel} \left[ \psi\left( \sum_{i \in S} x_i \right) \right] 
    \leq \mathbb{E}_{S \sim \mathcal{R}_{k,n}} \left[ \psi\left( \sum_{i \in S} x_i \right) \right].
\end{align*}
\end{lemma}

To apply \Cref{lem:hoeffdingexpectation}, we first establish the convexity of the Wasserstein distance.
\begin{lemma}[Convexity]\label{lem:w_convexity}
Let $\phi, \psi_1, \psi_2$ be discrete probability distributions over a metric space $(\mathcal{X}, d)$, and let $\psi = t \cdot \psi_1 + (1-t) \cdot \psi_2$ be their convex combination for some $t \in (0,1)$. Then, the following holds:  
\[
W(\phi, \psi) \le t \cdot W(\phi, \psi_1) + (1-t) \cdot  W(\phi, \psi_2).
\]
\end{lemma}
\begin{proof}
Let $\gamma_1$ and $\gamma_2$ be optimal couplings attaining the minimum in the definition of $W(\phi, \psi_1)$ and $W(\phi, \psi_2)$, respectively. That is,  
\[
\gamma_1 \in \argmin_{\gamma \in \Gamma(\phi, \psi_1)} \mathbb{E}_{(x,y) \sim \gamma} [d(x,y)]
\quad \text{and} \quad  
\gamma_2 \in \argmin_{\gamma \in \Gamma(\phi, \psi_2)} \mathbb{E}_{(x,y) \sim \gamma} [d(x,y)].
\]
Consider a coupling $\gamma=t \cdot \gamma_1+(1-t) \cdot \gamma_2$. It is easy to see that $\gamma \in \Gamma(\phi, \psi)$ since, given any $A \subseteq \mathcal{X}$, it holds that:
\begin{align*}
    \gamma(\mathcal{X} \times A) = t \cdot \phi(A) + (1-t) \cdot \phi(A) = \phi(A), \quad
    \gamma(A \times \mathcal{X}) = t\cdot \psi_1(A) + (1-t) \cdot \psi_2(A) = \psi(A).
\end{align*}
It follows that:
\begin{align*}
   \mathbb{E}_{(x,y) \sim \gamma}\left[ d(x,y) \right] &= t \cdot \mathbb{E}_{(x,y) \sim \gamma_1}\left[ d(x,y) \right] + (1-t) \cdot \mathbb{E}_{(x,y) \sim \gamma_2}\left[ d(x,y) \right]  = t \cdot W(\phi, \psi_1) + (1-t) \cdot  W(\phi, \psi_2).
\end{align*}
which concludes the proof.
\end{proof}

We are now ready to prove \Cref{lem:W1expectation}.

\SamplingWithoutReplacementExpectation*
\begin{proof}
Let $F = \{f_1, \ldots, f_n\}$, and consider the normed vector space $(\mathbb{R}^F, \|\cdot\|_1)$. 
We interpret any $v \in \mathbb{R}^F$ satisfying $v_j \geq 0$ for all $j \in F$ and $\sum_{j \in F} v_j = 1$ as a probability distribution over $F$, assigning mass $v_j$ to the element $j \in F$.
We define $\psi : \mathbb{R}^F \to \mathbb{R} \cup \{+\infty\}$ by:
\begin{align*}
    \psi(v) = \begin{cases}
        W(\phi_f^{[n]}, v) & \text{if } v_j \geq 0 \text{ for all } j \in F \text{ and } \sum_{j \in F} v_j = 1,\\
        +\infty & \text{otherwise.}
    \end{cases}
\end{align*}
Next, for each $i \in [n]$ we define the vector $x_i \in \mathbb{R}^F$ by setting for each $j \in F$:
\begin{align*}
    x_{i,j} = \begin{cases}
        1/k & \text{if } j = f(i), \\
        0 & \text{otherwise.}
    \end{cases}
\end{align*}
Then, by \Cref{lem:hoeffdingexpectation}, we obtain:
\begin{align*}
    \mathbb{E}_{S \sim \uniformpanel} \left[ W(\phi_f^{[n]}, \phi_f^S) \right] =  \mathbb{E}_{S \sim \uniformpanel} \left[ \psi\left( \sum_{i \in S} x_i \right) \right] \leq \mathbb{E}_{S \sim \mathcal{R}_{k,n}} \left[ \psi\left( \sum_{i \in S} x_i \right) \right]
    = \mathbb{E}_{S \sim \mathcal{R}_{k,n}} \left[ W(\phi_f^{[n]}, \phi_f^{S}) \right].
\end{align*}
as needed.
\end{proof}

\subsection{Proof of Lemma 2.4}
The following lemma follows directly from the proof of \citet*[Proposition 20]{weed2017sharpasymptoticfinitesamplerates}. 
\begin{lemma}[\citet*{weed2017sharpasymptoticfinitesamplerates}]\label{lem:weed}
   Let $f: [n] \to \mathcal{X}$ be a feature.
   Suppose that $d(x,y) \leq 1$ for all $x,y \in \mathcal{X}$.
    Let $x_1, \ldots, x_k \in \mathcal{X}$ and $x_i' \in \mathcal{X}$. Then, it holds that:
    \begin{align*}
        W(\phi_f^{[n]}, \phi_f^{(x_i, x_{-i})}) -  W(\phi_f^{[n]}, \phi_f^{(x_i', x_{-i})}) \leq 1/k.
    \end{align*}
\end{lemma}

We also use of the following McDiarmid-style inequality of \citet*[Proposition 2]{sambale2022concentration}.
\begin{lemma}[\citet*{sambale2022concentration}]\label{lem:sambale}
Let $g : [n]^k \to \mathbb{R}$ be a symmetric function; that is, for any permutation $\pi : [k] \to [k]$, we have $g(x_1, \ldots, x_k) = g(x_{\pi(1)}, \ldots, x_{\pi(k)})$.
Suppose there exists  $c > 0$ such that for all $x_1, \ldots, x_k, x_i' \in [n]$, it holds that $|g(x_i, x_{-i}) - g(x_i', x_{-i})| \leq c$.
Then, we have:
\begin{align*}
     \mathbb{P}_{S \sim \uniformpanel} \left[ g(x_S) \geq \mu + t \right] 
    \leq \exp\left( \frac{-t^2}{4 (1-k/n) c^2 k }\right) \quad\text{where }\mu = \mathbb{E}_{S \sim \uniformpanel}\left[g(x_S)\right].
\end{align*}
\end{lemma}

We are now ready to prove \Cref{lem:W1concentration}.

\SamplingWithoutReplacementConcentration*
\begin{proof}
Since $W(\phi_f^{[n]}, \cdot)$ is a symmetric function, and by \Cref{lem:weed}, we may apply \Cref{lem:sambale} with $c = 1/k$ to obtain:
\begin{align*}
    \mathbb{P}_{S \sim \uniformpanel} \left[ W(\phi_f^{[n]}, \phi_f^S) \geq \mu + t \right] \leq \exp\left( \frac{-t^2}{4 (1-k/n) c^2 k }\right) =  \exp\left( \frac{-t^2k}{4 (1-k/n) }\right) \leq  \exp\left( \frac{-t^2k}{4 }\right)
\end{align*}
as needed.
\end{proof}

\subsection{Proof of Lemma 2.6}

\WassersteinImpliesAvg*
\begin{proof}
Let $\phi^{[n]}_f$ and $\phi^S_f$ denote the population and panel distributions with respect to $f$, respectively, and let $f(A) \subseteq [0,1]$ represent the image of a subset $A \subseteq [n]$. 
The result follows from the Kantorovich–Rubinstein duality~\cite{villani2008optimal}, which provides an alternative expression for the Wasserstein distance:
\begin{equation*}
        W(\phi^{[n]}_f, \phi^S_f) = \sup_{||g||_{\mathsf{Lip}} \le 1} \left \{ \sum_{y \in f([n])} g(y) \cdot \phi^{[n]}_f(y) - \sum_{z \in f(S)} g(z) \cdot \phi^S_f(z) \right \},
\end{equation*}
where the supremum is taken over all Lipschitz functions $g : [0,1] \rightarrow \mathbb{R}$ with Lipschitz constant at most $1$, i.e., $||g||_{\mathsf{Lip}} \le 1$. The functions $g_1(y) = y$ and $g_2(y) = -y$ clearly satisfy this condition. Together with the assumption of $\epsilon$-representativeness, this implies:
\begin{align*}
       \epsilon \geq W(\phi^{[n]}_f, \phi^S_f) &\ge \sum_{y \in f([n])} y \cdot \phi^{[n]}_f(y) - \sum_{z \in f(S)} z \cdot \phi^S_f(z) = \frac{1}{n}\sum_{i \in[n]} f(i) - \frac{1}{|S|}\sum_{j \in S} f(j),\\
       \epsilon \geq W(\phi^{[n]}_f, \phi^S_f) &\ge \sum_{z \in f(S)} z \cdot \phi^S_f(z) - \sum_{y \in f([n])} y \cdot \phi^{[n]}_f(y) = \frac{1}{|S|}\sum_{j \in S} f(j) - \frac{1}{n}\sum_{i \in[n]} f(i).
\end{align*}
The result follows.
\end{proof}

\subsection{Proof of Lemma 2.8}

\panelcamouflaged*
\begin{proof}
Let us define $\mathcal{R}_{k,n}$ as the probability distribution corresponding to sampling \emph{with} replacement from the population.
We use the following result of \citet*[Lemma 5.10]{duetting2025pseudodimensioncontracts}, which guarantees if it holds that:
\begin{align*}
    \mathbb{P}_{(i_1, \ldots, i_k) \sim \mathcal{R}_{k,n}} \left[ \| g(f_{i_1, \ldots, i_k}(z,h,w,r)) - z \|_1 \leq  h/4 \right] \geq 3/4
\quad \text{for all } z \in \{-1,+1\}^{h},
\end{align*}
then it must be that $k = \Omega(h \cdot w^2)$.
A key distinction between this statement and that of the lemma is how the panel is formed. In the lemma, the panel consists of $k$ uniformly sampled agents without repetition, ensuring that set $S$ contains distinct elements. In contrast, \citet*{duetting2025pseudodimensioncontracts} considers drawing $k$ agents uniformly at random, allowing for possible repetitions. 

 We choose an integer parameter $r$ large enough such that $k^2/(2n) = k^2/(4 hwr) \leq 1/8$.

    Next, we define the event that no repetition occurs:
    \begin{align*}
        E = \{(i_1, \ldots, i_k) \in [n]^k \,:\, i_a \neq i_b \text{ for all }a,b \in [k]\}.
    \end{align*}
   We bound the probability of repetition using the union bound: 
    \begin{align*}
        \mathbb{P}_{(i_1,\ldots,i_k) \sim \mathcal{R}_{k,n}} [E] \geq 1- \sum_{a \neq b} \mathbb{P}_{(i_1,\ldots,i_k) \sim \mathcal{R}_{k,n}} [i_a = i_b] = 1-  {\binom{k}{2}} \cdot \frac{1}{n} \geq  1- \frac{k^2}{2n} \geq \frac{7}{8}.
    \end{align*}
    Thus, for all $z \in \{-1,+1\}^{\ell}$, we have:
\begin{align*}
   \mathbb{P}_{(i_1,\ldots,i_k) \sim \mathcal{R}_{k,n}} &\left[ \| g(f_{i_1, \ldots, i_k}(z,h,w,r)) - z \|_1 \leq h/4 \right] \\
   &\geq \mathbb{P}_{(i_1,\ldots,i_k) \sim \mathcal{R}_{k,n}} \left[ \| g(f_{i_1, \ldots, i_k}(z,h,w,r)) - z \|_1 \leq h/4 \text{ and } E \text{ holds}\right] \\
   &\geq \mathbb{P}_{(i_1,\ldots,i_k) \sim \mathcal{R}_{k,n}} \left[ \| g(f_{i_1, \ldots, i_k}(z,h,w,r)) - z \|_1 \leq h/4  \,|\, E\right] \cdot \mathbb{P}_{(i_1,\ldots,i_k) \sim \mathcal{R}_{k,n}} [E] \\
   &\geq \mathbb{P}_{(i_1,\ldots,i_k) \sim \mathcal{R}_{k,n}} \left[ \| g(f_{i_1, \ldots, i_k}(z,h,w,r)) - z \|_1 \leq h/4  \,|\, E\right] \cdot (7/8) \\
   &= \mathbb{P}_{S \sim \uniformpanel} \left[ \| g(f_S(z,h,w,r)) - z \|_1 \leq h/4 \right] \cdot (7/8) \\
   &\geq (6/7) \cdot (7/8) \\
   &\geq 3/4.
\end{align*}
This concludes the proof.
\end{proof}

\section{Missing Proofs of Section 3}\label{sec:missingfour}

\pbimpossibilityefficiency*
\begin{proof}
    Let $B_1, B_2 \geq 0$ satisfy $B_1 + B_2 = B$ and $B_1 + B_2 / m \leq 1$. When $B \geq 1$, choosing $B_1 = (m-B)/(m-1)$ and $B_2 = (m \cdot B - m) / (m-1)$ ensures that $B_1 + B_2 = (m-B+m\cdot B - m)/(m-1) = B$ and $B_1 + B_2/m = (m-B + B-1)/(m-1) = 1$. When $B < 1$, setting $B_1 = B$ and $B_2 = 0$ satisfies the conditions.
    For any $x \in [0,1]^m$, define three cost functions as:
    \begin{align*}
        \cost^{(0)}(x) &= (1/B_2) \cdot \sum_{i=1}^{m} \max(B_2/m - x_i, 0), \\
        \cost^{(1)}(x) &= B_1 + B_2/m - x_1, \\
        \cost^{(2)}(x) &= B_1 + B_2/m - x_2.
    \end{align*}
    These are valid cost functions in our model since $B_2 \leq m$ and $B_1 + B_2 / m \leq 1$.

Consider the following two participatory budgeting instances: $\langle m, B, (\cost^{(0)}, \ldots, \cost^{(0)}, \cost^{(1)}) \rangle$ and $\langle m, B, (\cost^{(0)}, \ldots, \cost^{(0)}, \cost^{(2)}) \rangle$. These instances differ only in the cost function assigned to the $n$-th agent.
In the first instance, an allocation $x$ with $x_i = B_2 / m$ for $i \neq 1$ and $x_1 = B_1 + B_2/m$ results in zero cost for all agents. 
This is a feasible budget allocation since $B_1 + B_2 = B$.
Similarly, in the second instance, setting $x_i = B_2 / m$ for $i \neq 2$ and $x_2 = B_1 + B_2/m$ also yields a cost of zero for all agents. Thus, the optimal cost in both instances is $\socialopt = 0$.

Now, let $y = \widetilde{x}(\cost^{(0)}, \ldots, \cost^{(0)})$. Since $\sum_{j=1}^{m} y_j \leq B$ and $B_1 + B_2 = B$, at least one of the following must hold: $y_1 < B_1 + B_2/m$, $y_2 < B_1 + B_2/m$, or $y_j < B_2/m$ for some $j \geq 3$.

In the first case, we obtain 
\begin{align*}
\mathbb{E}_{S \sim \uniformpanel} [ \social(\widetilde{x}(S)) ] \geq \mathbb{P}_{S \sim \uniformpanel}[S \subseteq \{1, \ldots, n-1\}] \cdot (B_1 + B_2/m - y_1) > 0.    
\end{align*}
The cases where $y_2 < B_1 + B_2/m$ or $y_j < B_2/m$ follow analogously.
\end{proof}

\section{Missing Proofs of Section 4}\label{apx:proofsfacilities}
In this section, we report all proofs and details missing from the discussion of the results in Section~\ref{sec:facility}.

\subsection{Proofs of Section 4.2}\label{sec:starsandthings}

\MultiplicativePACGuarantee*

\begin{proof}[Proof of~\Cref{thm:reduction}]
    We analyze the following formulation of the expected panel cost:
    \begin{align*}
        \mathbb{E}_{S \sim \uniformpanel}[\social(\overline{q}(S)) ] &= \int_{0}^\infty \mathbb{P}_{S \sim \uniformpanel}[\social(\overline{q}(S)) \geq x ] \,\mathrm{d}x.
    \end{align*}
    We split the analysis of this integral into different intervals of $x\in [0,\infty)$.
    We denote $\opt = \socialopt$.
    First, since probabilities are always bounded by $1$:
    \begin{align*}
        \int_{0}^{(1+\epsilon) \cdot \opt} \mathbb{P}_{S \sim \uniformpanel}[\social(\overline{q}(S)) \geq x ] \,\mathrm{d}x \leq (1 + \epsilon) \cdot \opt.
    \end{align*}
    Second, by the PAC-style guarantee with multiplicative error:
    \begin{align*}
        \int_{(1+\epsilon) \cdot \opt}^{4 \cdot \opt} \mathbb{P}_{S \sim \uniformpanel}[\social(\overline{q}(S)) \geq x ] \,\mathrm{d}x \leq \int_{(1+\epsilon) \cdot \opt}^{4 \cdot \opt} \epsilon \,\mathrm{d}x  \leq 3\epsilon \cdot \opt.
    \end{align*}
    Third, for $x > 4\cdot \opt$, we prepare to apply \Cref{thm:tail_bound}. To do so, observe that by triangle inequality 
    \[
    d(q^{\star},\overline{q}(S)) + \social(q^\star) =   d(q^{\star},\overline{q}(S)) + \opt \geq \social(\overline{q}(S)),
    \]
    implying for any $X>0$,
    \begin{align*}
        \mathbb{P}_{S \sim \uniformpanel}[\social(\overline{q}(S)) \geq (1+X)\cdot \opt ] \leq \mathbb{P}_{S \sim \uniformpanel}[d(q^\star, \overline{q}(S)) \geq X \cdot \opt].
    \end{align*}
    This gives us
    \begin{align*}
        \int_{4 \cdot \opt}^{\infty} &\mathbb{P}_{S \sim \uniformpanel}[\social(\overline{q}(S)) \geq x ] \,\mathrm{d}x \\
        & =  \int_{4 \cdot \opt}^{\infty} \mathbb{P}_{S \sim \uniformpanel}[\social(\overline{q}(S)) \geq (1 + (x/\opt - 1)) \cdot \opt ] \,\mathrm{d}x \\
        & \leq \int_{4\cdot \opt}^\infty \mathbb{P}_{S \sim \uniformpanel}[d(q^\star, \overline{q}(S)) \geq (x/\opt - 1) \cdot \opt] \ \mathrm{d}x.
    \end{align*}
For $x \geq 4 \cdot \opt$, we wish to apply \Cref{thm:tail_bound} with parameters $T = x/\opt - 1$ and $\delta = {3\epsilon}/{T^2}$. First, observe that $T > 2$ since in this case $T \geq 4 \cdot \opt/{\opt} - 1 = 3$.
We then compute the required panel size as:
    \[
    \frac{2 \cdot \log{(1/\delta)}}{\log \left(T^2/(4(T-1))\right)} = \frac{\log\left(T^2/(3\epsilon)\right)}{\log(T/8)}=  \frac{2\log(T) + \log\left(1/(3\epsilon))\right)}{\log(T) - \log(8)} = O\left(\log\left(\frac{1}{\epsilon}\right)\right).
    \]
Thus,  \Cref{thm:tail_bound} gives:
    \begin{align*}
       \int_{4\cdot \opt}^\infty &\mathbb{P}_{S \sim \uniformpanel}[d(q^\star, \overline{q}(S)) \geq (x/\opt - 1) \cdot \opt] \ \mathrm{d}x 
        \leq  \int_{4 \cdot \opt}^{\infty} \frac{3\epsilon}{(x/\opt - 1)^2} \,\mathrm{d}x  \\
        &=  3\epsilon \cdot \frac{\opt^2}{\opt-x} \Big|_{4 \cdot \opt}^{\infty} = 
3 \epsilon \cdot \frac{\opt}{3} = \epsilon \cdot \opt, 
    \end{align*}
as desired.
\end{proof}

\StarLowerBound*

\begin{proof}
Consider a population of $n = 2k + 1$ agents located at $x_1, \ldots, x_{k} = 0$ and $x_{k+1}, \ldots, x_{2k} = 1$.
We leave the location of $x_{2k + 1}$ unspecified for now.
Let $\mathcal{C} = \mathcal{X} = [0,1]$ and $d(x,y) = |x-y|$.
Notice that the unique optimal location is $q^\star = x_{2k + 1}$, regardless of $x_{2k+1}$. We compute: \begin{align*} \opt \leq \left({T}/2 + \left(1 - {T}/{2}\right)\right) \cdot (1/2) = 1/2. \end{align*}

Notice that the probability that a panel $S \sim \mathcal{U}_{k,n}$ contains agent $2k + 1$ is at most $1/2$.
Conditioned on the event that the panel does not contain agent $2k + 1$, we have that $\widetilde{x}(S)$ selects either $0$ with probability at least $1/2$ or $1$ with probability at least $1/2$.
Assume without loss of generality that the latter case holds, and let $x_{2k + 1} = 0$.
Now, with probability at least $1/4$ over the draw of a panel $S \sim \mathcal{U}_{k,n}$, the panel decision function $\widetilde{q}$ selects the facility at $1$. For each such panel $S$ we have:
\begin{align*} d(\widetilde{q}(S), q^\star) = d(1, 0) = 1 \geq 2 \cdot \opt, \end{align*} since $\opt \leq 1/2$ as shown above. This concludes the proof.
\end{proof}

\subsection{Proofs of Section 4.3} \label{apx:facility_efficiency}

\AssouadLowerBound*

\begin{proof}
   Our goal is to construct a function $g$ that satisfies the assumption of \Cref{lem:sample_lower_bound}, establishing a necessary lower bound on $k$.
    Let $w = \lfloor (1/64) \cdot (1 / \epsilon) \rfloor$. 
    Note that $\epsilon \leq (1/64) \cdot (1/w)$
    {and that $1/w = \Omega(1/\epsilon)$}.
    Let $r > 1$ be {an integer parameter} large enough to satisfy the conditions of \Cref{lem:sample_lower_bound}. 
    
    For every $z \in \{-1,+1\}^{t}$, consider a camouflaged population of size $n = 2twr$ with a feature $f(z,t,w,r)$ as in \Cref{def:almost_uniform}.
    We will now define a facility location instances $\mathcal{I}(z) = \langle \mathcal{X}, \mathcal{C}, d, (x_1^{(z)}, \ldots, x_n^{(z)}) \rangle$ as follows.
    Let $c = (1/2, 1/2, \ldots, 1/2) \in \mathcal{X}$ and let $e_j = (0, \ldots, 0, 1/2, 0, \ldots, 0)$ for $j \in [t]$ where the $1/2$ entry is in the $j$-th position.
    For every agent $i \in [n]$, we set:
   \begin{align*}
        x_i^{(z)} = \begin{cases}
            c + e_j & \text{if } f_i(z,t,r,w) = 2j,\\
            c - e_j & \text{if } f_i(z,t,r,w) = 2j-1.
        \end{cases}
    \end{align*}
    Note that exactly a $(1+z_j/w)/(2t)$-fraction of the population has location $c + e_j$, while exactly a $(1-z_j/w)/(2t)$-fraction has location  $c - e_j$. 
    It follows that:
    \begin{align*}
         \social(q) = \sum_{j=1}^t \left(\frac{1+z_j/w}{2t} \right) \cdot 
         \| c + e_j - q \|_{\infty}
         + \left( \frac{1-z_j/w}{2t} \right) \cdot \| c - e_j - q \|_{\infty}.
    \end{align*}
To obtain an upper bound on the optimal social cost, fix any $z \in \{-1, +1\}^t$ and consider the location:
\begin{align*}
    q^\star = c + \sum_{j \in [t]} (1/2) \cdot z_j \cdot e_j \in \mathcal{C}.  
    \end{align*}
    We obtain:
    \begin{align*}
       \socialopt(\mathcal{I}(z)) \leq \social(q^\star)
       = \sum_{j=1}^t \left(\frac{1+1/w}{2t} \right) \cdot (1/4) + \left( \frac{1-1/w}{2t} \right) \cdot (3/4) = \frac{1}{2} - \frac{1/w}{4}.
    \end{align*}
    Now, fix an arbitrary candidate location $q \in \mathcal{C}$. We aim to derive a concrete lower bound on $\social(q)$.
First, for any $j \in [t]$, since $1/4 \leq q_j \leq 3/4$, we have:     \begin{align*}
        \left(\frac{1+z_j/w}{2t} \right) \cdot \| c + e_j - q \|_{\infty} + \left( \frac{1-z_j/w}{2t} \right) \cdot  \| c - e_j - q \|_{\infty} \geq   \frac{1}{2t} - \frac{1/w}{4t}.
    \end{align*}
    Furthermore, if $q_j \leq 1/2$ and $z_j = +1$, then:  
    \begin{align*}
         \left(\frac{1+z_j/w}{2t} \right) \cdot \| c + e_j - q \|_{\infty} + \left( \frac{1-z_j/w}{2t} \right) \cdot  \| c - e_j - q \|_{\infty}  &\geq  \left(\frac{1+1/w}{2t}\right) \cdot (1-q_j) + \left( \frac{1-1/w}{2t} \right) \cdot q_j \\
        &= \frac{1}{2t} + \frac{(1/w) \cdot (1-2q_j)}{2t}  \\
        &\geq \frac{1}{2t}.
    \end{align*}
Similarly, if $q_j \geq 1/2$ and $z_j = -1$, then:
    \begin{align*}
         \left(\frac{1+z_j/w}{2t} \right) \cdot \| c + e_j - q \|_{\infty} + \left( \frac{1-z_j/w}{2t} \right) \cdot  \| c - e_j - q \|_{\infty}  &\geq  \left(\frac{1+1/w}{2t}\right) \cdot q_j + \left( \frac{1-1/w}{2t} \right) \cdot (1-q_j) \\
        &= \frac{1}{2t} + \frac{(1/w) \cdot (2q_j-1)}{2t}  \\
        &\geq \frac{1}{2t}.
    \end{align*}
     Define $g : \mathcal{X}^k \to \{-1,+1\}^{t}$ as follows: 
\begin{align*}
    g_j(x_1, \ldots, x_k) =
    \begin{cases}
        +1 & \text{if } \widetilde{q}_{j}(x_1, \ldots, x_k) > 1/2, \\
        -1 & \text{otherwise}.
    \end{cases}
\end{align*}
Note that $g$ is independent of $z$.
We use $g(f_S(z,t,w,r))$ for $S \subseteq [n]$ to denote $g((x_i^{(z)})_{i \in S})$.

Let $S \subseteq [n]$ be such that  
$\|g(f_S(z,t,w,r)) - z\|_1 > (1/4) \cdot t$. From the analysis above, we get:  
\begin{align*}
    \social(\widetilde{q}(S))- \socialopt(\mathcal{I}(z))  \geq \left(\frac{t}{4} \right) \cdot \left( 
\frac{1/w}{4t} \right) = \frac{1}{16w} \geq \frac{1}{8w} \cdot \socialopt(\mathcal{I}(z))
\end{align*}
Combining this with our assumption on $\overline{q}(S)$, and the fact that $\epsilon \leq (1/64) \cdot (1/w)$, we obtain:
\begin{align*}
    \frac{1}{64w} \cdot \socialopt(\mathcal{I}(z)) &\geq \epsilon \cdot \socialopt(\mathcal{I}(z)) \\ &\geq \mathbb{E}_{S \sim \uniformpanel} \left[\social (\widetilde{q}(S)) - \socialopt(\mathcal{I}(z)) \right] \\
    &> \frac{1}{8w} \cdot \socialopt(\mathcal{I}(z))\cdot \mathbb{P}_{S \sim \uniformpanel} \left[ \|g(f_S(z,t,w,r)) - z\|_1 > (1/4) \cdot t \right] 
\end{align*}
It follows that:
\[
\mathbb{P}_{S \sim \uniformpanel} \left[ \|g(f_S(z,t,w,r)) - z\|_1 \leq (1/4) \cdot t \right] \geq 7/8 \quad \text{for all $z \in \{-1,+1\}^{t}$}.
\]
Hence, $g$ satisfies the assumption of \Cref{lem:sample_lower_bound}, which concludes the proof.
\end{proof}

\subsection{Proofs of Section 4.4}\label{apx:multifacilities}

Next, we will prove the panel complexity upper bound, based on the lemma above.
\PanelComplexityMultifacilities*
\begin{proof}
Let us again consider the feature $f(i) = x_i$ for $i \in [n]$.
We adopt the following notation:
\begin{align*}
       \mu_f =  \mathbb{E}_{S \sim \uniformpanel}\left[W(\phi_f^{[n]}, \phi_f^S)\right] \quad \text{and}\quad  \mu'_f =  \mathbb{E}_{S \sim \mathcal{R}_{k,n}}\left[W(\phi_f^{[n]}, \phi_f^S)\right].
    \end{align*}
Applying the result of \citet*[Proposition 2.1]{chewi2024statisticaloptimaltransport}, we can choose $k = O(1/\epsilon^2)$ such that $\mu'_f \leq \epsilon$.
Then, by \Cref{lem:W1expectation}, it follows that $\mu_f \leq \epsilon$.
Finally, by \Cref{lem:multifacilities}, we obtain:
\begin{equation*}
    \mathbb{E}_{S \sim \uniformpanel} \left[ \social(\overline{y}(S)) \right] \leq \socialopt + \mu_f \leq \socialopt + \epsilon
\end{equation*}
as needed.
\end{proof}

\end{document}